\documentclass[11pt,letterpaper,reqno]{amsart}

\title[Symplectic Semiclassical Dynamics]{Symplectic Semiclassical Wave Packet Dynamics}
\author{Tomoki Ohsawa}
\address{Department of Mathematics \& Statistics, University of Michigan--Dearborn, 4901 Evergreen Road, Dearborn, MI 48128-2406}
\author{Melvin Leok}
\address{Department of Mathematics, University of California, San Diego, 9500 Gilman Dr, La Jolla, CA 92093-0112}
\email{ohsawa@umich.edu, mleok@math.ucsd.edu}
\date{\today}



\usepackage{graphicx,mathrsfs,subfigure}
\usepackage[margin=1in, marginpar=.5in]{geometry}

\usepackage{tikz-cd}

\usepackage[numbers,sort&compress]{natbib}

\usepackage[colorlinks=false]{hyperref}

\theoremstyle{plain}
\newtheorem{theorem}{Theorem}[section]

\newtheorem{proposition}[theorem]{Proposition}

\theoremstyle{definition}

\theoremstyle{remark}
\newtheorem{remark}[theorem]{Remark}

\def\od#1#2{\dfrac{d#1}{d#2}}
\def\pd#1#2{\dfrac{\partial #1}{\partial #2}}

\def\tod#1#2{d#1/d#2}

\def\parentheses#1{\!\left(#1\right)}
\def\brackets#1{\!\left[#1\right]}
\def\braces#1{\!\left\{#1\right\}}

\def\tr{\mathop{\mathrm{tr}}\nolimits}

\def\norm#1{\left\|#1\right\|}

\def\DS{\displaystyle}
\def\R{\mathbb{R}}
\def\C{\mathbb{C}}

\def\defeq{\mathrel{\mathop:}=}
\def\eqdef{=\mathrel{\mathop:}}
\def\setdef#1#2{ \left\{ #1 \ |\ #2 \right\} }
\def\ip#1#2{\left\langle#1,#2\right\rangle}
\def\exval#1{\left\langle#1\right\rangle}
\def\texval#1{\langle#1\rangle}

\def\ket#1{|#1\rangle}
\def\diag{\operatorname{diag}}

\renewcommand{\Re}{\operatorname{Re}}
\renewcommand{\Im}{\operatorname{Im}}
\def\eps{\varepsilon}

\def\d{{\bf d}}
\def\i{{\bf i}}

\begin{document}

\footskip=.6in

\begin{abstract}
  The paper gives a symplectic-geometric account of semiclassical Gaussian wave packet dynamics.
  We employ geometric techniques to ``strip away'' the symplectic structure behind the time-dependent Schr\"odinger equation and incorporate it into semiclassical wave packet dynamics.
  We show that the Gaussian wave packet dynamics is a Hamiltonian system with respect to the symplectic structure, apply the theory of symplectic reduction and reconstruction to the dynamics, and discuss dynamic and geometric phases in semiclassical mechanics.
  A simple harmonic oscillator example is worked out to illustrate the results: We show that the reduced semiclassical harmonic oscillator dynamics is completely integrable by finding the action--angle coordinates for the system, and calculate the associated dynamic and geometric phases explicitly.
  We also propose an asymptotic approximation of the potential term that provides a practical semiclassical correction term to the approximation by Heller.
  Numerical results for a simple one-dimensional example show that the semiclassical correction term realizes a semiclassical tunneling.
\end{abstract}

\maketitle

\section{Introduction}
\subsection{Background}
Gaussian wave packet dynamics is an essential example in time-dependent semiclassical mechanics that nicely illustrates the classical--quantum correspondence, as well as a widely-used tool in simulations of semiclassical mechanics, particularly in chemical physics~(see, e.g., \citet{Ta2007} and \citet{Lu2008}).
A Gaussian wave packet is a particular form of wave function whose motion is governed by a trajectory of a classical ``particle''; hence it provides an explicit connection between classical and quantum dynamics by placing ``(quantum mechanical) wave flesh on classical bones.''~\cite{BeMo1972,Ta2007}

The most remarkable feature of Gaussian wave packet dynamics is that, for quadratic potentials, the Gaussian wave packet is known to give an {\em exact} solution of the Schr\"odinger equation if and only if the underlying ``particle'' dynamics satisfies a certain set of ordinary differential equations.
Even with non-quadratic potentials, Gaussian wave packet dynamics is an effective tool to approximate the full quantum dynamics, as demonstrated by, among others, a series of works by \citet{He1975a,He1976b,He1981} and \citet{Ha1980, Ha1998}.
See also \citet{RuSm2013} for a use of the Gaussian wave packets to transform the Schr\"odinger equation into more computationally tractable equations in the semiclassical regime.

One popular approach to semiclassical dynamics is the use of propagators obtained by semiclassical approximations of Feynman's path integral~\cite{FeHi2010}.
Whereas the original work of \citet{He1975a} does not involve the path integral, a number of methods have been developed by applying these propagators to Gaussian wave packets to derive the time evolution of semiclassical systems (see, e.g., \citet{Heller-LesHouches}, \citet{Gr2006}, \citet[Chapter~10]{Ta2007} and references therein).

On the other hand, it also turns out that Gaussian wave packet dynamics has nice geometric structures associated with it.
\citet{An1988b,An1990,An1991} showed that the {\em frozen} Gaussian wave packet dynamics inherits symplectic and Riemannian structures from quantum mechanics.
\citet{FaLu2006} (see also \citet[Section~II.4]{Lu2008}) found the symplectic/Poisson structure of the {\em``thawed''} spherical Gaussian wave packet dynamics (which is more general than the frozen one) and developed a numerical integrator that preserve the geometric structure.
It is worth noting that \citet{He1975a} decouples the classical and quantum parts of the dynamics and only recognizes the classical part as a Hamiltonian system, whereas \citet{FaLu2006} show that the whole system is Hamiltonian.

\subsection{Main Results and Outline}
The main contribution of the present paper is to provide a symplectic and Hamiltonian view of Gaussian wave packet dynamics.
Our main source of inspiration is the series of works by Lubich and his collaborators compiled in \citet{Lu2008}.
Much of the work here builds on or gives an alternative view of their results.
Our focus here is the symplectic point of view, as opposed to the mainly variational and Poisson ones of \citet{FaLu2006} and \citet{Lu2008}.
Also, our results give a multi-dimensional generalization of the work by \citet{PaSc1994} from a mathematical---mainly geometric---point of view.

In Section~\ref{sec:SymplecticModelReduction}, we start with a review of some key results in \cite{Lu2008} from the symplectic point of view, and then consider the non-spherical Gaussian wave packet dynamics in Section~\ref{sec:Non-Spherical}.
The main result in Section~\ref{sec:Non-Spherical} shows that the non-spherical Gaussian wave packet dynamics is a Hamiltonian system with respect to the symplectic structure found by a technique outlined in Section~\ref{sec:SymplecticModelReduction}; the result is shown to specialize to the spherical case of \citet{FaLu2006} in Section~\ref{sec:Spherical}.
Then, in Section~\ref{sec:Reduction}, we exploit the symplectic point of view to discuss the symplectic reduction of the non-spherical Gaussian wave packet dynamics.
This naturally leads to the reconstruction of the full dynamics and the associated dynamic and geometric phases in Section~\ref{sec:ReconstructionAndGeometricPhase}.
Section~\ref{sec:Potential} gives an asymptotic analysis of the potential terms present in the Hamiltonian formulation.
The potential terms usually cannot be evaluated analytically and one may need to approximate them for practical applications.
We propose an asymptotic approximation that provides a correction term to the locally quadratic approximation of Heller.
Finally, we consider two simple examples: the semiclassical harmonic oscillator in Sections~\ref{sec:SemiclassicalHO} and a semiclassical tunneling in \ref{sec:SemiclassicalTunneling}.
The semiclassical harmonic oscillator is completely integrable: We find action--angle coordinates using the Darboux coordinates found in Section~\ref{sec:Spherical} and the associated Hamilton--Jacobi equation, and also find the explicit formula for the reconstruction phase.
The semiclassical tunneling example is solved numerically to demonstrate a classically forbidden motion of a semiclassical particle.

\section{Symplectic Model Reduction for Quantum Mechanics}
\label{sec:SymplecticModelReduction}
This section shows how one may reduce an infinite-dimensional quantum dynamics to a finite-dimensional semiclassical dynamics from the symplectic-geometric point of view.
It will also be shown that the finite-dimensional dynamics defined below is optimal in the sense described in Section~\ref{ssec:metrics}.
We follow \citet[Chapter~II]{Lu2008} with more emphasis on the geometric aspects to better understand the geometry behind the model reduction.

\subsection{Symplectic View of the Schr\"odinger Equation}
Let $\mathcal{H}$ be a complex (often infinite-dimensional) Hilbert space equipped with a (right-linear) inner product $\ip{\cdot}{\cdot}$.
It is well-known~(see, e.g., \citet[Section~2.2]{MaRa1999}) that the two-form $\Omega$ on $\mathcal{H}$ defined by
\begin{equation*}
  \Omega(\psi_{1}, \psi_{2}) = 2\hbar \Im\ip{\psi_{1}}{\psi_{2}}
\end{equation*}
is a symplectic form, and hence $\mathcal{H}$ is a symplectic vector space.
One may also define the one-form $\Theta$ on $\mathcal{H}$ by
\begin{equation*}
  \Theta(\psi) = -\hbar \Im\ip{\psi}{{\bf d}\psi};
  \quad
  \ip{\Theta(\psi)}{\varphi} = -\hbar \Im\ip{\psi}{\varphi}.
\end{equation*}
Then, one has $\Omega = -{\bf d}\Theta$.
Now, given a Hamiltonian operator\footnote{In general, the Hamiltonian operator $\hat{H}$ may not be defined on the whole $\mathcal{H}$.} $\hat{H}$ on $\mathcal{H}$, we may write the expectation value of the Hamiltonian $\texval{\hat{H}}: \mathcal{H} \to \R$ as
\begin{equation*}
 \texval{\hat{H}}(\psi) \defeq \texval{\psi, \hat{H}\psi}.
\end{equation*}
Then, the corresponding Hamiltonian flow
\begin{equation*}
  X_{\exval{\hat{H}}} = \dot{\psi}\,\pd{}{\psi}  
\end{equation*}
on $\mathcal{H}$ defined by
\begin{equation}
  \label{eq:SchroedingerEq}
  {\bf i}_{X_{\texval{\hat{H}}}} \Omega = {\bf d}\texval{\hat{H}}
\end{equation}
gives the Schr\"odinger equation
\begin{equation*}
  \dot{\psi} = -\frac{i}{\hbar}\hat{H} \psi.
\end{equation*}

\subsection{Symplectic Model Reduction}
Let $\mathcal{M}$ be a {\em finite}-dimensional manifold and suppose there exists an embedding $\iota: \mathcal{M} \hookrightarrow \mathcal{H}$ and hence $\iota(\mathcal{M})$ is a submanifold of $\mathcal{H}$.
\begin{proposition}[{\citet[Section~II.1]{Lu2008}}]
  \label{prop:embedding}
  If the manifold $\mathcal{M}$ is equipped with an almost complex structure $J_{y}: T_{y}\mathcal{M} \to T_{y}\mathcal{M}$ such that
  \begin{equation}
    \label{eq:AlmostComplexStructure}
    T_{y}\iota \circ J_{y} = i\cdot T_{y}\iota
  \end{equation}
  for any $y \in \mathcal{M}$, then $\mathcal{M}$ is a symplectic manifold with symplectic form $\Omega_{\mathcal{M}} \defeq \iota^{*}\Omega$.
\end{proposition}
The proof of \citet{Lu2008} is based on the projection from $\mathcal{H}$ to the tangent space $T_{\iota(y)}\iota(\mathcal{M})$ of the embedded manifold $\iota(\mathcal{M})$.
We give a proof from a slightly different perspective using the embedding $\iota: \mathcal{M} \hookrightarrow \mathcal{H}$ more explicitly.
As we shall see later, the embedding $\iota$ is the key ingredient exploited to define geometric structures on the semiclassical side as the pull-backs of the corresponding structures on the quantum side.
\begin{proof}
  It is easy to show that $\Omega_{\mathcal{M}}$ is closed: $\d\Omega_{\mathcal{M}} = \iota^{*}\d\Omega = 0$.
  We then need to show that $\Omega_{\mathcal{M}}$ is non-degenerate, i.e., $T_{y}\mathcal{M} \cap (T_{y}\mathcal{M})^{\perp} = \{0\}$, where $(\,\cdot\,)^{\perp}$ stands for the symplectic complement with respect to $\Omega_{\mathcal{M}}$.
  Let $v_{y} \in T_{y}\mathcal{M} \cap (T_{y}\mathcal{M})^{\perp}$; then $J_{y}(v_{y}) \in T_{y}\mathcal{M}$ and thus
  \begin{align*}
    0 &= \Omega_{\mathcal{M}}\parentheses{ v_{y}, J_{y}(v_{y}) }
    \\
    & = \Omega\parentheses{ T_{y}\iota(v_{y}), T_{y}\iota \circ J_{y}(v_{y}) }
    \\
    & = 2\hbar \Im\ip{ T_{y}\iota(v_{y}) }{ i\,T_{y}\iota(v_{y}) }
    \\
    & = 2\hbar \Re\ip{ T_{y}\iota(v_{y}) }{ T_{y}\iota(v_{y}) }
    \\
    & = 2\hbar \ip{ T_{y}\iota(v_{y}) }{ T_{y}\iota(v_{y}) }.
  \end{align*}
  Hence $T_{y}\iota(v_{y}) = 0$ and so $v_{y} = 0$ since $\iota$ is injective.
  Therefore, $T_{y}\mathcal{M} \cap (T_{y}\mathcal{M})^{\perp} = \{0\}$ and thus $\mathcal{M}$ is symplectic with the symplectic form $\Omega_{\mathcal{M}}$.
\end{proof}

Now, define a Hamiltonian $H: \mathcal{M} \to \R$ by the pull-back
\begin{equation*}
  H \defeq \iota^{*} \texval{\hat{H}} = \texval{\hat{H}} \circ \iota.
\end{equation*}
Then, we may define a Hamiltonian system on $\mathcal{M}$ by
\begin{equation}
  \label{eq:HamiltonianSystem-X_H}
  {\bf i}_{X_{H}} \Omega_{\mathcal{M}} = {\bf d}H.  
\end{equation}
Hence we ``reduced'' the infinite-dimensional Hamiltonian dynamics $X_{\texval{\hat{H}}}$ on $\mathcal{H}$ to the finite-dimensional Hamiltonian dynamics $X_{H}$ on $\mathcal{M}$.

\begin{remark}
  \label{rem:Dirac-Frenkel}
  One may also take a variational approach using the Dirac--Frenkel variational principle (see, e.g., \citet[Section~II.1]{Lu2008} and references therein) to derive \eqref{eq:HamiltonianSystem-X_H}; this is also a variational principle behind other time-dependent approximation methods such as the time-dependent Hartree--Fock method (see, e.g., \citet[Section~II.3]{Lu2008}).
\end{remark}

\begin{remark}
  The idea of restricting a Hamiltonian dynamics on a (pre-)symplectic manifold to a symplectic submanifold is reminiscent of the constraint algorithm of \citet{GoNeHi1978} and \citet{GoNe1979b, GoNe1980}.
  However, in our setting, both the original and restricted (or reduced) dynamics are defined on {\em strictly} symplectic (as opposed to pre-symplectic) manifolds and thus we do not need to resort to the constraint algorithm as long as the conditions in Proposition~\ref{prop:embedding} are satisfied.
\end{remark}

If we write the embedding $\iota: \mathcal{M} \hookrightarrow \mathcal{H}$ explicitly as $y \mapsto \chi(y)$, then one may first find a symplectic one-form $\Theta_{\mathcal{M}}$ on $\mathcal{M}$ as the pull-back of $\Theta$ by $\iota$, i.e.,
\begin{equation}
\label{eq:Theta_M-def}
  \Theta_{\mathcal{M}} \defeq \iota^{*}\Theta = -\hbar \Im\ip{ \chi }{ \pd{\chi}{y^{j}} } {\bf d}y^{j}.
\end{equation}
Then, the symplectic form $\Omega_{\mathcal{M}} \defeq \iota^{*}\Omega$ is given by
\begin{equation*}
  \Omega_{\mathcal{M}} = -{\bf d}\Theta_{\mathcal{M}}.
\end{equation*}
On the other hand, one can calculate the Hamiltonian $H: \mathcal{M} \to \R$ as follows:
\begin{equation}
  \label{eq:H-def}
  H(y) = \texval{\chi(y), \hat{H}\chi(y)}.
\end{equation}

\subsection{Riemannian Metrics and Least Squares Approximation}
\label{ssec:metrics}
As shown by \citet[Section~II.1.2]{Lu2008}, it turns out that the the finite-dimensional dynamics $X_{H}$ is the least squares approximation to the original dynamics $X_{\texval{\hat{H}}}$ in the sense we will describe below.
Again, \citet{Lu2008} exploits the projection from $\mathcal{H}$ to the tangent space $T_{\iota(y)}\iota(\mathcal{M})$, but we give an alternative account using the metrics naturally induced on $\mathcal{H}$ and $\mathcal{M}$.

First recall (see, e.g., \citet[Section~5.3]{MaRa1999} and \citet[Section~5.1.1]{ChJa2004}) that any complex Hilbert space $\mathcal{H}$ is equipped with a Riemannian metric naturally induced by its inner product.
In our setting, we may define
\begin{equation*}
  g(\psi_{1}, \psi_{2}) \defeq 2\hbar \Re\ip{\psi_{1}}{\psi_{2}}
\end{equation*}
so that it is compatible with the symplectic structure $\Omega$ in the sense that
\begin{equation}
  \label{eq:compatibility}
  g(i\psi_{1}, \psi_{2}) = \Omega(\psi_{1}, \psi_{2})
  \quad \text{and}\quad
  \Omega(\psi_{1}, i\psi_{2}) = g(\psi_{1}, \psi_{2}).
\end{equation}
Then, we may induce a metric on $\mathcal{M}$ by the pull-back
\begin{equation*}
 g_{\mathcal{M}} \defeq \iota^{*}g,
\end{equation*}
and thus we may define norms $\norm{\,\cdot\,}$ and $\norm{\,\cdot\,}_{\mathcal{M}}$ for tangent vectors on $\mathcal{H}$ and $\mathcal{M}$, respectively, as follows:
\begin{equation*}
  \norm{X} \defeq \sqrt{ g(X,X) },
  \qquad
  \norm{v}_{\mathcal{M}} \defeq \sqrt{ g_{\mathcal{M}}(v,v) }.
\end{equation*}

\begin{proposition}[{\citet[Section~II.1.2]{Lu2008}}]
  If the manifold $\mathcal{M}$ is equipped with an almost complex structure $J_{y}: T_{y}\mathcal{M} \to T_{y}\mathcal{M}$ that satisfies \eqref{eq:AlmostComplexStructure}, then the the Hamiltonian vector field $X_{H}$ on $\mathcal{M}$ defined by \eqref{eq:HamiltonianSystem-X_H} is the least squares approximation among the vector fields on $\mathcal{M}$ to the vector field $X_{\exval{\hat{H}}}$ defined by the Schr\"odinger equation~\eqref{eq:SchroedingerEq}: For any $y \in \mathcal{M}$ let $\eta \defeq \iota(y) \in \mathcal{H}$; then, for any $w_{y} \in T_{y}\mathcal{M}$,
  \begin{equation*}
    \| X_{\texval{\hat{H}}}(\eta) - T_{y}\iota(w_{y}) \|^{2} \ge \| X_{\texval{\hat{H}}}(\eta) - T_{y}\iota(X_H(y)) \|^{2}
    = \| X_{\texval{\hat{H}}}(\eta) \|^{2} - \| X_{H}(y) \|_{\mathcal{M}}^{2},
  \end{equation*}
  where the equality holds if and only if $w_{y} = X_{H}(y)$.
\end{proposition}

\begin{proof}
  Notice first that the inclusion map $\iota$ pulls back the compatible triple---metric, symplectic form, and complex structure---to $\mathcal{M}$, i.e., Eq.~\eqref{eq:compatibility} implies, for any $v, w \in T\mathcal{M}$, 
  \begin{equation*}
    g_{\mathcal{M}}(J(v), w) = \Omega_{\mathcal{M}}(v, w)
    \quad \text{and}\quad
    \Omega_{\mathcal{M}}(v, J(w)) = g_{\mathcal{M}}(v, w).
  \end{equation*}
  We may then estimate the difference between $X_{\texval{\hat{H}}}$ and $W \defeq T\iota(w)$ for any $w \in T\mathcal{M}$ as follows:
  \begin{align*}
    \| X_{\texval{\hat{H}}} - W \|^{2}
    &= g\parentheses{ X_{\texval{\hat{H}}} - W, X_{\texval{\hat{H}}} - W }
    \\
    &= g\parentheses{ X_{\texval{\hat{H}}}, X_{\texval{\hat{H}}} }
    - 2g\parentheses{ X_{\texval{\hat{H}}}, W }
    + g\parentheses{ W, W },
  \end{align*}
  where
  \begin{align*}
    g\parentheses{ X_{\texval{\hat{H}}}, W }
    &= \Omega\parentheses{ X_{\texval{\hat{H}}}, iW }
    \\
    &= \Omega\parentheses{ X_{\texval{\hat{H}}}, T\iota \circ J(w) }
    \\
    &= {\bf d}\texval{\hat{H}} \cdot T\iota \circ J(w)
    \\
    &= {\bf d}(\iota^{*} \texval{\hat{H}}) \cdot J(w)
    \\
    &= {\bf d}H \cdot J(w)
    \\
    &= \Omega_{\mathcal{M}}\parentheses{ X_{H}, J(w) }
    \\
    &= g_{\mathcal{M}}\parentheses{ X_{H}, w },
  \end{align*}
  and $g\parentheses{ W, W } = g_{\mathcal{M}}(w, w)$.
  Therefore,
  \begin{align*}
    \| X_{\texval{\hat{H}}} - W \|^{2}
    &= g\parentheses{ X_{\texval{\hat{H}}}, X_{\texval{\hat{H}}} }
    - 2g_{\mathcal{M}}\parentheses{ X_{H}, w }
    + g_{\mathcal{M}}(w, w)
    \\
    &= \| X_{\texval{\hat{H}}} \|^{2} - \| X_{H} \|_{\mathcal{M}}^{2}
    + g_{\mathcal{M}}\parentheses{ X_{H}-w , X_{H}-w }
    \\
    &\ge \| X_{\texval{\hat{H}}} \|^{2} - \| X_{H} \|_{\mathcal{M}}^{2},
  \end{align*}
  where the equality holds if and only if $w = X_{H}$.
\end{proof}

\section{Gaussian Wave Packet Dynamics}
\label{sec:Non-Spherical}
\subsection{Gaussian Wave Packets}
\label{ssec:GaussianWPDyn}
In particular, let $\mathcal{H} \defeq L^{2}(\R^{d})$ with the standard right-linear inner product $\ip{\cdot}{\cdot}$ and $\hat{H}$ be the Schr\"odinger operator:
\begin{equation*}
  \hat{H} = -\frac{\hbar^{2}}{2m} \Delta + V(x),
\end{equation*}
where $\Delta$ is the Laplacian in $\R^{d}$.

Let us now consider the following specific form of $\chi$ called the {\em (non-spherical) Gaussian wave packet}  (see, e.g., \citet{He1975a, He1976b}):
\begin{equation}
  \label{eq:chi}
  \chi(y;x) = \exp\braces{ \frac{i}{\hbar}\brackets{ \frac{1}{2}(x - q)^{T}\mathcal{C}(x - q) + p \cdot (x - q) + (\phi + i \delta) } },
\end{equation}
where $\mathcal{C} = \mathcal{A} + i\mathcal{B}$ is a $d \times d$ complex symmetric matrix with a positive-definite imaginary part, i.e., the matrix $\mathcal{C}$ is an element in the {\em Siegel upper half space}~\cite{Si1943} defined by
\begin{equation*}
  \Sigma_{d} \defeq 
  \setdef{ \mathcal{C} = \mathcal{A} + i\mathcal{B} \in \mathbb{C}^{d\times d} }{ \mathcal{A}, \mathcal{B} \in \text{Sym}_{d}(\R),\, \mathcal{B} > 0 },
\end{equation*}
where $\text{Sym}_{d}(\R)$ is the set of $d \times d$ real symmetric matrices, and $\mathcal{B} > 0$ means that $\mathcal{B}$ is positive-definite.
It is easy to see that the (real) dimension of $\Sigma_{d}$ is $d(d+1)$.

One may then let $\mathcal{M}$ be the $(d+1)(d+2)$-dimensional manifold
\begin{equation*}
 \mathcal{M} = T^{*}\R^{d} \times \Sigma_{d} \times \mathbb{S}^{1} \times \R,
\end{equation*}
and a typical element $y \in \mathcal{M}$ is written as follows:
\begin{equation*}
  y \defeq (q, p, \mathcal{A}, \mathcal{B}, \phi, \delta).
\end{equation*}
We then define an embedding of $\mathcal{M}$ to $\mathcal{H} \defeq L^{2}(\R^{d})$ by
\begin{equation*}
  \iota: \mathcal{M} \hookrightarrow \mathcal{H};
  \quad
  \iota(y) = \chi(y;\,\cdot\,)
\end{equation*}
with Eq.~\eqref{eq:chi}. Then, it is easy to show that the embedding $\iota: \mathcal{M} \hookrightarrow \mathcal{H}$ in fact satisfies condition \eqref{eq:AlmostComplexStructure} of Proposition~\ref{prop:embedding}, where the almost complex structure $J_{y}: T_{y}\mathcal{M} \to T_{y}\mathcal{M}$ is given by
\begin{multline*}
  J_{y}\parentheses{ \dot{q}, \dot{p}, \dot{\mathcal{A}}, \dot{\mathcal{B}}, \dot{\phi}, \dot{\delta} }
  \\
  = \parentheses{
    \mathcal{B}^{-1}(\mathcal{A}\dot{q} - \dot{p}),\,
    (\mathcal{A}\mathcal{B}^{-1}\mathcal{A} + \mathcal{B})\dot{q} - \mathcal{A}\mathcal{B}^{-1}\dot{p},\,
    -\dot{\mathcal{B}},\,
    \dot{\mathcal{A}},\,
    p^{T}\mathcal{B}^{-1}(\mathcal{A}\dot{q} - \dot{p}) - \dot{\delta},\,
    -p \cdot \dot{q} + \dot{\phi}
  },
\end{multline*}
and hence $\mathcal{M}$ is symplectic.

Note that the variable $\delta$ is essential in the symplectic formulation.
We have
\begin{equation}
  \label{eq:N}
  \mathcal{N}(\mathcal{B},\delta) \defeq \norm{\chi(y;\,\cdot\,)}^{2} = \sqrt{ \frac{(\pi\hbar)^{d}}{\det \mathcal{B}} }\, \exp\parentheses{ -\frac{2\delta}{\hbar} },
\end{equation}
and so we may eliminate $\delta$ by solving $\norm{\chi} = 1$ for $\delta$ and substituting it back into Eq.~\eqref{eq:chi} to normalize it.
However, {\em without $\delta$, the manifold $\mathcal{M}$ is odd-dimensional and hence cannot be symplectic.}
More specifically, {\em the variable $\delta$ plays the role of incorporating the phase variable $\phi$ into the symplectic setting.}
\begin{remark}
  \label{rem:normalization}
  As we shall see later, $\mathcal{N}(\mathcal{B},\delta) = \norm{\chi}^{2}$ is essentially the conserved quantity (momentum map) corresponding to a symmetry of the system (by Noether's theorem).
  Normalization is introduced as the restriction of $\chi$ to the level set $\norm{\chi} = 1$ of the conserved quantity, i.e., $\chi$ is normalized on the invariant submanifold of $\mathcal{M}$ defined by $\norm{\chi} = 1$.
  Furthermore, this setup naturally fits into the setting of symplectic reduction and reconstruction as we shall see in Sections~\ref{sec:Reduction} and \ref{sec:ReconstructionAndGeometricPhase}.
\end{remark}

\subsection{Symplectic Gaussian Wave Packet Dynamics}
We may now calculate the symplectic one-form $\Theta_{\mathcal{M}}$, Eq.~\eqref{eq:Theta_M-def}, explicitly as
\begin{equation}
  \label{eq:Theta_M}
  \Theta_{\mathcal{M}} \defeq \iota^{*}\Theta = \mathcal{N}(\mathcal{B},\delta) \parentheses{ p_{i}\,{\bf d}q^{i} - \frac{\hbar}{4}\tr(\mathcal{B}^{-1}{\bf d}\mathcal{A}) - {\bf d}\phi },
\end{equation}
and hence also the symplectic form on $\mathcal{M}$:
\begin{align}
  \label{eq:Omega_M}
  \Omega_{\mathcal{M}} &\defeq -{\bf d}\Theta_{\mathcal{M}}
  \nonumber\\
  &= \mathcal{N}(\mathcal{B},\delta) \biggl\{
  {\bf d}q^{i} \wedge {\bf d}p_{i} - \frac{p_{i}}{2} {\bf d}q^{i} \wedge \tr(\mathcal{B}^{-1}{\bf d}\mathcal{B}) - \frac{2p_{i}}{\hbar}{\bf d}q^{i} \wedge {\bf d}\delta
  \nonumber\\
  &\qquad\quad
  +\frac{\hbar}{8}(2\mathcal{B}^{-1}_{ik}\mathcal{B}^{-1}_{lj} + \mathcal{B}^{-1}_{ij}\mathcal{B}^{-1}_{lk})\d\mathcal{A}_{ij} \wedge \d\mathcal{B}_{kl}
  \nonumber\\
  &\qquad\quad
  + \frac{1}{2}\brackets{ \tr(\mathcal{B}^{-1}\d\mathcal{A}) \wedge \d\delta - \tr(\mathcal{B}^{-1}\d\mathcal{B}) \wedge \d\phi }
  + \frac{2}{\hbar} {\bf d}\phi \wedge {\bf d}\delta
  \biggr\}.
\end{align}
On the other hand, the Hamiltonian becomes
\begin{align}
  H &= \mathcal{N}(\mathcal{B},\delta) \braces{ \frac{p^{2}}{2m} + \frac{\hbar}{4m}\tr\brackets{ \mathcal{B}^{-1}(\mathcal{A}^{2} + \mathcal{B}^{2}) } } + \exval{V}(q, \mathcal{B}, \delta)
  \nonumber\\
  &= \mathcal{N}(\mathcal{B},\delta) \braces{
    \frac{p^{2}}{2m} + \frac{\hbar}{4m}\tr\brackets{ \mathcal{B}^{-1}(\mathcal{A}^{2} + \mathcal{B}^{2}) } + \overline{\exval{V}}(q, \mathcal{B})
  },
  \label{eq:H}
\end{align}
where $\exval{V}(q, \mathcal{B}, \delta)$ is the expectation value of the potential $V$ for the above wave function $\chi$, i.e.,
\begin{equation*}
  \exval{V}(q, \mathcal{B}, \delta)
  \defeq \exp\parentheses{ -\frac{2\delta}{\hbar} } \int_{\R^{d}} V(x) \exp\brackets{ -\frac{1}{\hbar} (x - q)^{T}\mathcal{B}(x - q) } dx
\end{equation*}
and $\overline{\exval{V}}(q, \mathcal{B})$ is a normalized version of it:
\begin{equation}
  \label{eq:avgV}
  \overline{\exval{V}}(q, \mathcal{B})
  \defeq \frac{\exval{V}(q, \mathcal{B}, \delta)}{\mathcal{N}(\mathcal{B},\delta)}
  = \sqrt{ \frac{\det \mathcal{B}}{(\pi\hbar)^{d}} } \int_{\R^{d}} V(x) \exp\brackets{ -\frac{1}{\hbar}(x - q)^{T}\mathcal{B}(x - q) } dx.
\end{equation}
In what follows, for any function $A(x)$ such that $\exval{A} < \infty$, we write
\begin{equation*}
  \overline{ \exval{A} } \defeq \frac{\exval{A}}{\mathcal{N}(\mathcal{B},\delta)}
  = \ip{ \frac{\chi}{\norm{\chi}} }{ A\,\frac{\chi}{\norm{\chi}} }.
\end{equation*}
Note that if $\chi$ is normalized, i.e., $\mathcal{N}(\mathcal{B},\delta) = \norm{\chi}^{2} = 1$, then $\overline{ \exval{A} } = \exval{A}$; in particular $\overline{\exval{V}} = \exval{V}$.

Now, the main result in this section is the following:
\begin{theorem}
  The Hamiltonian system ${\bf i}_{X_{H}} \Omega_{\mathcal{M}} = {\bf d}H$ with the above symplectic form~\eqref{eq:Omega_M} and Hamiltonian~\eqref{eq:H} gives the semiclassical equations (see also \citet[Section~II.4.1]{Lu2008}):
  \begin{equation}
    \label{eq:Heller}
    \begin{array}{c}
      \DS
      \dot{q} = \frac{p}{m},
      \qquad
      \dot{p} = -\overline{ \exval{\nabla{V}} },
      \qquad
      \dot{\mathcal{A}} = -\frac{1}{m}(\mathcal{A}^{2} - \mathcal{B}^{2}) - \overline{ \exval{\nabla^{2}V} },
      \qquad
      \dot{\mathcal{B}} = -\frac{1}{m}(\mathcal{A}\mathcal{B} + \mathcal{B}\mathcal{A}),
      \medskip\\
      \DS
      \dot{\phi} = \frac{p^{2}}{2m} - \overline{ \exval{V} }
      - \frac{\hbar}{2m} \tr\mathcal{B}
      + \frac{\hbar}{4} \tr\parentheses{ \mathcal{B}^{-1} \overline{ \exval{\nabla^{2}V} } },
      \qquad
      \dot{\delta} = \frac{\hbar}{2m} \tr\mathcal{A},
    \end{array}
  \end{equation}
  where $\nabla^{2}V$ is the $d \times d$ Hessian matrix, i.e.,
  \begin{equation*}
    (\nabla^{2}V)_{ij} = \pd{^{2}V}{x^{i}\partial x^{j}}.
  \end{equation*}
\end{theorem}

\begin{proof}
  Calculation of ${\bf i}_{X_{H}} \Omega_{\mathcal{M}}$ is straightforward, whereas that of $\d{H}$ is somewhat tedious:
  Note first that the derivatives of the potential term $\overline{ \exval{V} }(q, \mathcal{B})$ are rewritten as follows using integration by parts:
  \begin{equation*}
    \pd{}{q} \overline{ \exval{V} } = \overline{\exval{\nabla{V}}},
    \qquad
    \pd{}{\mathcal{B}_{ij}} \overline{ \exval{V} } = -\frac{\hbar}{4} \parentheses{ \mathcal{B}^{-1} \overline{\exval{\nabla^{2}{V}}} \mathcal{B}^{-1} }_{ij}.
  \end{equation*}
  As a result, we have
  \begin{multline*}
    \d{H} = \mathcal{N}(q,\mathcal{B}) \biggl(
    \overline{\exval{\nabla{V}}} \cdot \d{q} + \frac{p}{m} \cdot \d{p}
    + \frac{\hbar}{4m} \tr\brackets{ (\mathcal{A}\mathcal{B}^{-1} + \mathcal{B}^{-1}\mathcal{A})\,\d\mathcal{A} }
    \\
    + \frac{\hbar}{4} \tr\braces{
      \brackets{
        \frac{1}{m} (I_{d} - \mathcal{B}^{-1}\mathcal{A}^{2}\mathcal{B}^{-1})
        - \frac{2}{\hbar}\,\overline{H} \mathcal{B}^{-1}
        - \mathcal{B}^{-1} \overline{\exval{\nabla^{2}{V}}} \mathcal{B}^{-1}
      } \d\mathcal{B}
    }
    - \frac{2}{\hbar}\,\overline{H}\,\d{\delta}
    \biggr),
  \end{multline*}
  where $I_{d}$ is the identity matrix of size $d$ and $\overline{H}$ is what later appears as the reduced Hamiltonian in Eq.~\eqref{eq:H-reduced}:
  \begin{equation*}
    \overline{H} \defeq \frac{p^{2}}{2m} + \frac{\hbar}{4m}\tr\brackets{ \mathcal{B}^{-1}(\mathcal{A}^{2} + \mathcal{B}^{2}) } + \overline{\exval{V}}(q, \mathcal{B}). \qedhere
  \end{equation*}
\end{proof}

\begin{remark}
  Writing $\mathcal{C} = \mathcal{A} + i\mathcal{B}$, the above equations for $\mathcal{A}$ and $\mathcal{B}$ are combined into the following single equation:
  \begin{equation*}
    \dot{\mathcal{C}} = -\frac{1}{m}\mathcal{C}^{2} - \overline{ \exval{\nabla^{2}V} }.
  \end{equation*}
\end{remark}

\begin{remark}
  Approximation of solutions of the Schr\"odinger equation~\eqref{eq:SchroedingerEq} by the Gaussian wave packet~\eqref{eq:chi} with the semiclassical equations~\eqref{eq:Heller} is usually valid for short-times.
  Specifically, \citet[Theorem~4.4]{Lu2008} estimates that the error $\norm{ \chi(y(t);x) - \psi(x,t) }$ is $O(t\sqrt{\hbar})$.
  See \citet{Ha1980} for a similar but more detailed result.
\end{remark}

\begin{remark}
  \label{rem:avgV}
  The original formulation of \citet{He1975a} (see also \citet{LeHe1982}) is {\em not} from a Hamiltonian/symplectic point of view and does not involve expectation values $\overline{\exval{V}}$ etc.
  The above equations seem to be originally derived in \citet{CoKa1990} by using the Dirac--Frenkel variational principle (see Remark~\ref{rem:Dirac-Frenkel}); its Hamiltonian structure for the reduced dynamics (see Theorem~\ref{thm:reduction}) in the one-dimensional case was discovered in \citet{PaSc1994} by finding Darboux coordinates (see Remark~\ref{rem:DarbouxCoords}) explicitly.
  Its connection with the symplectic structure for the full quantum dynamics is elucidated in \citet{FaLu2006} for the spherical Gaussian wave packets (see Section~\ref{sec:Spherical}) and for a general abstract case in \citet[Section~II.1]{Lu2008}, which is restated in Proposition~\ref{prop:embedding}.
\end{remark}

\subsection{Relationship with Alternative Approach using Time-Dependent Operators}
\label{ssec:Littlejohn}
There is an alternative approach, due to \citet[Section~7]{Li1986}, to deriving time-evolution equations for the Gaussian wave packet~\eqref{eq:chi}.
The key idea behind it is to describe the dynamics in terms of time-dependent {\em operators} acting on the initial state, as opposed to assuming, from the outset, a wave function containing time-dependent parameters as in \eqref{eq:chi}:
Let $\ket{\psi_{0}}$ be the initial state and suppose that the state at the time $t$, $\ket{\psi(t)}$, is given by
\begin{equation}
  \label{eq:LittlejohnWavePacket}
  \ket{\psi(t)} = e^{i\phi(t)/\hbar}\, T(q(t),p(t))\, M(S(t))\, T(q_{0},p_{0})^{*}\, \ket{\psi_{0}},
\end{equation}
where $T(\delta q, \delta p)$ is the Heisenberg operator corresponding to the translation $(q,p) \mapsto (q + \delta q, p + \delta p)$ in $T^{*}\R^{d}$ (see \citet[Section~3]{Li1986}); $S(t) \in Sp(2d,\R)$ and $M(S(t))$ is a corresponding metaplectic operator (see \citet[Section~4]{Li1986}); $q_{0}$ and $p_{0}$ are expectation values of the standard position and momentum operators for the initial state $\ket{\psi_{0}}$.

One finds a connection with the Gaussian wave packet~\eqref{eq:chi} by choosing the ground state of the harmonic oscillator as the initial state $\ket{\psi_{0}}$, i.e.,
\begin{equation*}
  \psi_{0}(x) \defeq \langle x\,|\,\psi_{0} \rangle = \frac{1}{(\pi\hbar)^{d/4}} \exp\parentheses{ -\frac{|x|^{2}}{2\hbar} }.
\end{equation*}
Then, one obtains the ``ground state'' of the wave packets of \citet{Ha1980, Ha1998} (see also \citet[Chapter~V]{Lu2008}):
\begin{align}
  \psi(x,t) &\defeq \langle x\,|\,\psi(t) \rangle
  \nonumber\\
  & = (\pi\hbar)^{-d/4} |\det Q|^{-1/2} \exp\braces{ \frac{i}{\hbar}\brackets{ \frac{1}{2}(x - q)^{T}P Q^{-1}(x - q) + p \cdot (x - q) + \phi } },
  \label{eq:HagedornWavePacket}
\end{align}
where the parameters $(q,p,Q,P,\phi)$ are time $t$ dependent, but this is suppressed for brevity; the $d \times d$ complex matrices $Q$ and $P$ are introduced by writing $S \in Sp(2d,\R)$ as
\begin{equation*}
  S = 
  \begin{bmatrix}
    \Re Q & \Im Q \smallskip\\
    \Re P & \Im P
  \end{bmatrix}.
\end{equation*}

It turns out that the above wave packet~\eqref{eq:HagedornWavePacket} is a normalized version of \eqref{eq:chi} (up to some difference in the phase $\phi$) if $S \in Sp(2d,\R)$ and $\mathcal{A} + i\mathcal{B} \in \Sigma_{d}$ are related by
\begin{equation*}
  \pi_{U(d)}(S) = P Q^{-1} = \mathcal{A} + i\mathcal{B} 
\end{equation*}
where $\pi_{U(d)}$ is the quotient map defined as
\begin{equation*}
  \pi_{U(d)}: Sp(2d,\R) \to \Sigma_{d};
  \quad
  \begin{bmatrix}
    A & B \\
    C & D
  \end{bmatrix}
  \mapsto
  (C + iD)(A + iB)^{-1},
\end{equation*}
which naturally arises by identifying $\Sigma_{d}$ as the homogeneous space $Sp(2d,\R)/U(d)$ (see \citet{Si1943}, \citet[Section~4.5]{Fo1989}, and \citet[Exercise~2.28 on p.~48]{McSa1999}).
We also note that \citet[Section~8.1]{Li1986} exploits the identification $\Sigma_{d} \cong Sp(2d,\R)/U(d)$ to parametrize Wigner functions of Gaussian wave packets.

\citet[Section~7]{Li1986} derives the dynamics for the parameters $(q,p,Q,P,\phi)$ by substituting \eqref{eq:LittlejohnWavePacket} into the Schr\"odinger equation \eqref{eq:SchroedingerEq} with its Hamiltonian operator being approximated by an operator that is quadratic in the standard position and momentum operators: More specifically, one first calculates the quadratic approximation of the Weyl symbol of the original Hamiltonian, and then obtains the corresponding operator by inverting the Weyl symbol relations.

The advantage of this approach is that one may choose an arbitrary initial state for $\ket{\psi_{0}}$ and hence is more general than assuming the Gaussian wave packet \eqref{eq:chi}.
However, {\em the resulting equations (see (7.25) of \cite{Li1986}) for $(q, p)$ are classical Hamilton's equations as in those of \citet{He1975a, He1976b}, whereas the second equation of \eqref{eq:Heller} has the potential term $\overline{ \exval{\nabla{V}} }(q,\mathcal{B})$, which generally depends on $\mathcal{B}$ and hence contains a quantum correction}.
The $\mathcal{B}$-dependence of the potential term is crucial for us because it allows the system to realize classically forbidden motions such as tunneling (see Section~\ref{sec:SemiclassicalTunneling}).

\section{Momentum Map, Normalization, and Symplectic Reduction}
\label{sec:Reduction}
The previous section showed that the symplectic structure for the semiclassical dynamics~\eqref{eq:Heller} is inherited from the one for the Schr\"odinger equation by pull-back via the inclusion $\iota: \mathcal{M} \to \mathcal{H}$.
In this section, we show that the semiclassical dynamics also inherits the phase symmetry and the corresponding momentum map from the (full) quantum dynamics, and thus we may perform symplectic reduction, as is done for the Schr\"odinger equation in \citet[Section~5A]{MaMoRa1990} and \citet[Section~6.3]{Ma1992}.

\subsection{Geometry of Quantum Mechanics}
Consider the $\mathbb{S}^{1}$-action $\Psi: \mathbb{S}^{1} \times \mathcal{H} \to \mathcal{H}$ on the Hilbert space $\mathcal{H} = L^{2}(\R^{d})$ defined by
\begin{equation*}
  \Psi_{\theta}: \mathcal{H} \to \mathcal{H};
  \quad
  \psi \mapsto e^{i \theta}\psi.
\end{equation*}
The corresponding momentum map ${\bf J}: \mathcal{H} \to \mathfrak{so}(2)^{*} \cong \R$, where we identified $\mathbb{S}^{1}$ with $SO(2)$, is given by~(see, e.g., \citet[Section~6.3]{Ma1992})
\begin{equation*}
  {\bf J}(\psi) = -\hbar\norm{\psi}^{2}.
\end{equation*}
The expectation value of the Hamiltonian $\texval{\hat{H}}$ is invariant under this action, and hence Noether's theorem implies that the norm $\norm{\psi}$ is conserved along the flow of the Schr\"odinger equation.
In particular, the level set at the value $-\hbar$ gives the unit sphere $\mathbb{S}(\mathcal{H})$ in the Hilbert space $\mathcal{H}$, i.e., the set of normalized wave functions:
\begin{equation*}
  {\bf J}^{-1}(-\hbar) \defeq \setdef{ \psi \in \mathcal{H} }{ \norm{\psi} = 1 } \eqdef \mathbb{S}(\mathcal{H}).
\end{equation*}
Since $\mathbb{S}^{1}$ is Abelian, the projective Hilbert space $\mathbb{P}(\mathcal{H}) = {\bf J}^{-1}(-\hbar)/\mathbb{S}^{1} = \mathbb{S}(\mathcal{H})/\mathbb{S}^{1}$ is the reduced space in Marsden--Weinstein reduction~\cite{MaWe1974} and hence is symplectic: Defining an inclusion $\hat{i}_{\hbar}$ and projection $\hat{\pi}_{\hbar}$ by
\begin{equation*}
  \hat{i}_{\hbar}: {\bf J}^{-1}(-\hbar) \hookrightarrow \mathcal{H},
  \qquad
  \hat{\pi}_{\hbar}: {\bf J}^{-1}(-\hbar) \to \mathbb{P}(\mathcal{H}),
\end{equation*}
we have the symplectic form $\overline{\Omega}$ on $\mathbb{P}(\mathcal{H})$ such that
\begin{equation*}
 \hat{\pi}_{\hbar}^{*} \overline{\Omega} = \hat{i}_{\hbar}^{*} \Omega.
\end{equation*}
We may then reduce the dynamics to $\mathbb{P}(\mathcal{H})$.
Note that the geometric phase (Aharonov--Anandan phase~\cite{AhAn1987}) arises naturally as a reconstruction phase, as shown in \citet[Section~5A]{MaMoRa1990} and \citet[Section~6.3]{Ma1992}.

\subsection{Geometry of Gaussian Wave Packet Dynamics}
\label{ssec:GeometryOfGWPDyn}
The geometry and dynamics in $\mathcal{M}$ inherit this setting as follows:
Define an $\mathbb{S}^{1}$-action $\Phi: \mathbb{S}^{1} \times \mathcal{M} \to \mathcal{M}$ on the manifold $\mathcal{M}$ by
\begin{equation*}
  \Phi_{\theta}: \mathcal{M} \to \mathcal{M};
  \quad
  (q, p, \mathcal{A}, \mathcal{B}, \phi, \delta) \mapsto (a, p, \mathcal{A}, \mathcal{B}, \phi + \hbar\,\theta, \delta).
\end{equation*}
Then, it is clear that the diagram below commutes, and hence $\Phi$ is the $\mathbb{S}^{1}$-action on $\mathcal{M}$ induced by the action $\Psi$ on $\mathcal{H}$.
\begin{equation*}
  \begin{tikzcd}
    \mathcal{M} \arrow[hook]{r}{\iota} \arrow{d}[swap]{\Phi_{\theta}} & \mathcal{H} \arrow{d}{\Psi_{\theta}}
    \\
    \mathcal{M} \arrow[hook]{r}{\iota} & \mathcal{H}
  \end{tikzcd}
\end{equation*}
The infinitesimal generator of the action with $\xi \in \mathfrak{so}(2) \cong \R$ is 
\begin{equation*}
  \xi_{\mathcal{M}}(y) \defeq \left. \od{}{\eps} \Phi_{\eps\xi}(y) \right|_{\eps=0}
  = \hbar\,\xi\,\pd{}{\phi}.
\end{equation*}
The corresponding momentum map ${\bf J}_{\!\mathcal{M}}: \mathcal{M} \to \mathfrak{so}(2)^{*} \cong \R$ is defined by the condition
\begin{equation*}
  \ip{ {\bf J}_{\!\mathcal{M}}(y) }{\xi} = \ip{ \Theta_{\mathcal{M}}(y) }{ \xi_{\mathcal{M}}(y) } = -\hbar\,\mathcal{N}(\mathcal{B},\delta)\,\xi,
\end{equation*}
for any $\xi \in \mathfrak{so}(2)$ and hence
\begin{equation*}
  {\bf J}_{\!\mathcal{M}}(y) = -\hbar\,\mathcal{N}(\mathcal{B},\delta).
\end{equation*}
Thus, we see that ${\bf J}_{\!\mathcal{M}} = {\bf J} \circ \iota$ or  ${\bf J}_{\!\mathcal{M}}(y) = {\bf J}(\chi(y))$.

Now, the Hamiltonian $H: \mathcal{M} \to \R$ is invariant under the action, and hence again by Noether's theorem, ${\bf J}_{\mathcal{M}}$ is conserved along the flow of $X_{H}$, i.e., each level set of ${\bf J}_{\mathcal{M}}$ is an invariant submanifold of the dynamics $X_{H}$.
In particular, on the level set
\begin{equation*}
  {\bf J}_{\!\mathcal{M}}^{-1}(-\hbar) \defeq \setdef{ y \in \mathcal{M} }{ {\bf J}_{\!\mathcal{M}}(y) = -\hbar },
\end{equation*}
we have $\mathcal{N}(\mathcal{B},\delta) = 1$ and thus, by Eq.~\eqref{eq:N}, the Gaussian wave packet function $\chi$ is normalized, i.e. $\norm{\chi} = 1$, and we may write
\begin{equation*}
  \chi|_{{\bf J}_{\!\mathcal{M}}^{-1}(-\hbar)}(x) = \parentheses{ \frac{\det\mathcal{B}}{(\pi\hbar)^{d}} }^{1/4} \exp\braces{ \frac{i}{\hbar}\brackets{ \frac{1}{2}(x - q)^{T}(\mathcal{A} + i\mathcal{B})(x - q) + p \cdot (x - q) + \phi } }
\end{equation*}
by eliminating the variable $\delta$ as alluded in Section~\ref{ssec:GaussianWPDyn}.
Ignoring the phase factor $e^{i\phi/\hbar}$ in the above expression corresponds to taking the equivalence class defined by the $\mathbb{S}^{1}$-action, and so the wave function
\begin{equation*}
  \parentheses{ \frac{\det\mathcal{B}}{(\pi\hbar)^{d}} }^{1/4} \exp\braces{ \frac{i}{\hbar}\brackets{ \frac{1}{2}(x - q)^{T}(\mathcal{A} + i\mathcal{B})(x - q) + p \cdot (x - q) } }
\end{equation*}
may be thought of as a representative for the equivalence class $[ \chi|_{{\bf J}_{\!\mathcal{M}}^{-1}(-\hbar)} ]$ in the projective Hilbert space $\mathbb{P}(\mathcal{H})$.

\begin{theorem}[Reduction of Gaussian wave packet dynamics]
  \label{thm:reduction}
  The semiclassical Hamiltonian system~\eqref{eq:Heller} on $\mathcal{M}$ is reduced by the above $\mathbb{S}^{1}$-symmetry to the Hamiltonian system
  \begin{equation}
    \label{eq:HamiltonianSystem-X_h}
    {\bf i}_{X_{\overline{H}}} \overline{\Omega}_{\hbar} = {\bf d}\overline{H}
  \end{equation}
  defined on
  \begin{equation*}
    \overline{\mathcal{M}}_{\hbar} \defeq {\bf J}_{\!\mathcal{M}}^{-1}(-\hbar)/\mathbb{S}^{1} = T^{*}\R^{d} \times \Sigma_{d},
  \end{equation*}
  with the reduced symplectic form
  \begin{equation}
    \label{eq:Omega-reduced}
    \overline{\Omega}_{\hbar} = {\bf d}q^{i} \wedge {\bf d}p_{i} + \frac{\hbar}{4} \mathcal{B}^{-1}_{ik} \mathcal{B}^{-1}_{lj} \d\mathcal{A}_{ij} \wedge \d\mathcal{B}_{kl}
  \end{equation}
  and the reduced Hamiltonian
  \begin{equation}
    \label{eq:H-reduced}
    \overline{H} = \frac{p^{2}}{2m} + \frac{\hbar}{4m}\tr\brackets{ \mathcal{B}^{-1}(\mathcal{A}^{2} + \mathcal{B}^{2}) } + \overline{\exval{V}}(q, \mathcal{B}).
  \end{equation}
  As a result, Eq.~\eqref{eq:HamiltonianSystem-X_h} gives the reduced set of the semiclassical equations:
  \begin{equation}
    \label{eq:Heller-reduced}
    \dot{q} = \frac{p}{m},
    \qquad
    \dot{p} = -\overline{\exval{ \nabla{V} }},
    \qquad
    \dot{\mathcal{A}} = -\frac{1}{m}(\mathcal{A}^{2} - \mathcal{B}^{2}) - \overline{\exval{ \nabla^{2}V }},
    \qquad
    \dot{\mathcal{B}} = -\frac{1}{m}(\mathcal{A}\mathcal{B} + \mathcal{B}\mathcal{A}).
  \end{equation}
\end{theorem}

A few remarks are in order before the proof:
\begin{remark}
  Note that the reduced symplectic form $\overline{\Omega}_{\hbar}$ is much simpler than the original one $\Omega_{\mathcal{M}}$ in Eq.~\eqref{eq:Omega_M}; it consists of the canonical symplectic form of classical mechanics and a ``quantum'' term proportional to $\hbar$.
  The quantum term is in fact essentially the imaginary part of the Hermitian metric
  \begin{equation*}
    g_{\Sigma_{d}} \defeq \tr\parentheses{ \mathcal{B}^{-1} \d\mathcal{C}\,\mathcal{B}^{-1} \d\bar{\mathcal{C}}\, }
    = \mathcal{B}^{-1}_{ik} \mathcal{B}^{-1}_{lj} \d\mathcal{C}_{kl} \otimes \d\bar{\mathcal{C}}_{ij}
  \end{equation*}
  on the Siegel upper half space $\Sigma_{d}$~\cite{Si1943}, i.e.,
  \begin{equation*}
    \Im g_{\Sigma_{d}} = -\mathcal{B}^{-1}_{ik} \mathcal{B}^{-1}_{lj} \d\mathcal{A}_{ij} \wedge \d\mathcal{B}_{kl},
  \end{equation*}
  and this gives a symplectic structure on the Siegel upper half space $\Sigma_{d}$.
\end{remark}

\begin{remark}
  Again, we may replace the last two equations of \eqref{eq:Heller-reduced} by the succinct form
  \begin{equation*}
    \dot{\mathcal{C}} = -\frac{1}{m}\mathcal{C}^{2} - \overline{\exval{ \nabla^{2}V }}
  \end{equation*}
  with $\mathcal{C} = \mathcal{A} + i\mathcal{B}$.
\end{remark}

\begin{proof}[Proof of Theorem~\ref{thm:reduction}]
  A simple application of Marsden--Weinstein reduction~\cite{MaWe1974} (see also \citet[Sections~1.1 and 1.2]{MaMiOrPeRa2007}).
  In fact, all the geometric ingredients necessary for the reduction are inherited from the (full) quantum dynamics as follows:
  Define the inclusion
  \begin{equation*}
    i_{\hbar}: {\bf J}_{\!\mathcal{M}}^{-1}(-\hbar) \hookrightarrow \mathcal{M},
  \end{equation*}
  the quotient map
  \begin{equation*}
    \pi_{\hbar}: {\bf J}_{\!\mathcal{M}}^{-1}(-\hbar) \to {\bf J}_{\!\mathcal{M}}^{-1}(-\hbar)/\mathbb{S}^{1} \eqdef \overline{\mathcal{M}}_{\hbar},
  \end{equation*}
  and also another inclusion
  \begin{equation*}
    [\iota]: \overline{\mathcal{M}}_{\hbar} \to \mathbb{P}(\mathcal{H});
    \quad
    [y] \mapsto [\chi(y)],
  \end{equation*}
  where $[\,\cdot\,]$ stands for the equivalence classes defined by the $\mathbb{S}^{1}$-actions $\Psi$ and $\Phi$.
  Then, the diagram below commutes and shows how the geometric structures are pulled back to the semiclassical side.
  \begin{equation*}
    \begin{tikzcd}[column sep=8.5ex, row sep=6ex]
      \mathcal{M} \arrow[hook]{r}{\iota} & \mathcal{H}
      \\
      {\bf J}_{\!\mathcal{M}}^{-1}(-\hbar) \arrow[hook]{r}[swap]{\iota|_{{\bf J}_{\!\mathcal{M}}^{-1}(-\hbar)}} \arrow{u}{i_{\hbar}} \arrow{d}[swap]{\pi_{\hbar}} & {\bf J}^{-1}(-\hbar) \arrow{u}[swap]{\hat{i}_{\hbar}} \arrow{d}{\hat{\pi}_{\hbar}}
      \\
      \overline{\mathcal{M}}_{\hbar} \arrow[hook]{r}[swap]{[\iota]} & \mathbb{P}(\mathcal{H})
    \end{tikzcd}
  \end{equation*}
  Figure~\ref{fig:Geometry} gives a schematic of the inheritance.
  \begin{figure}[htbp]
    \centering
    \includegraphics[width=.75\linewidth]{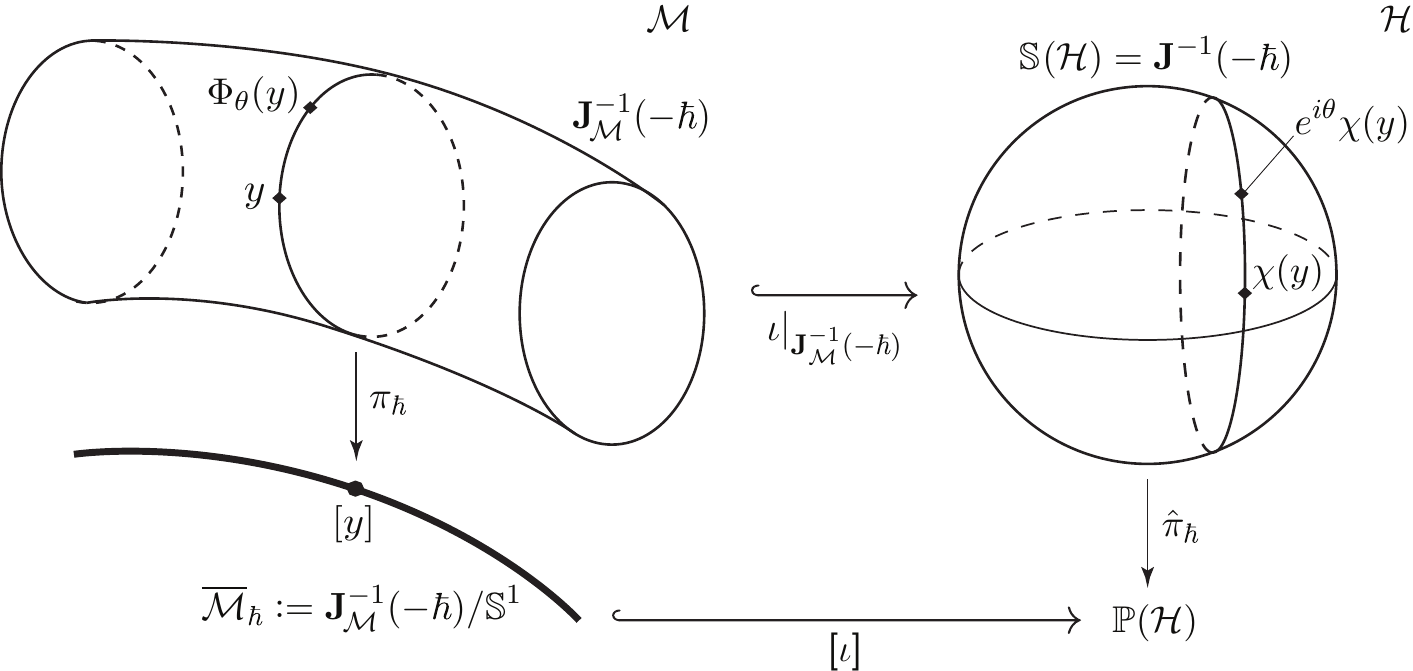}
    \caption{Geometry of Gaussian wave packet dynamics: The geometric structures necessary for symplectic reduction of semiclassical dynamics on $\mathcal{M}$ are inherited from the full quantum dynamics in $\mathcal{H}$ as pull-backs by inclusions.}
    \label{fig:Geometry}
  \end{figure}

  The level set ${\bf J}_{\!\mathcal{M}}^{-1}(-\hbar)$ is defined by $\mathcal{N}(\mathcal{B},\delta) = 1$, and so one may eliminate $\delta$ (see Eq.~\eqref{eq:N}) to write
  \begin{equation*}
    {\bf J}_{\!\mathcal{M}}^{-1}(-\hbar) = T^{*}\R^{d} \times \Sigma_{d} \times \mathbb{S}^{1}  = \{ (q, p, \mathcal{A}, \mathcal{B}, \phi) \},
  \end{equation*}
  and therefore the Marsden--Weinstein quotient is given by
  \begin{equation*}
    \overline{\mathcal{M}}_{\hbar} \defeq {\bf J}_{\!\mathcal{M}}^{-1}(-\hbar)/\mathbb{S}^{1} = T^{*}\R^{d} \times \Sigma_{d} = \{ (q, p, \mathcal{A}, \mathcal{B}) \}.
  \end{equation*}
  Then, the reduced symplectic form \eqref{eq:Omega-reduced} follows from coordinate calculations using its defining relation
  \begin{equation*}
    \pi_{\hbar}^{*}\, \overline{\Omega}_{\hbar} = i_{\hbar}^{*}\, \Omega_{\mathcal{M}}.
  \end{equation*}

  We also have the reduced Hamiltonian $\overline{H}: \overline{\mathcal{M}}_{\hbar} \to \R$, which appeared earlier in Eq.~\eqref{eq:H-reduced}, uniquely defined by
  \begin{equation*}
    \overline{H} \circ \pi_{\hbar} = H|_{{\bf J}_{\!\mathcal{M}}^{-1}(-\hbar)}
  \end{equation*}
  due to the $\mathbb{S}^{1}$-invariance of the original Hamiltonian $H$.

  Then, the Hamiltonian dynamics $\i_{X_{H}}\Omega_{\mathcal{M}} = \d{H}$ on $\mathcal{M}$ is reduced to the Hamiltonian dynamics ${\bf i}_{X_{\overline{H}}} \overline{\Omega}_{\hbar} = {\bf d}\overline{H}$ on the reduced space $\overline{\mathcal{M}}_{\hbar}$.
\end{proof}

\section{Spherical Gaussian Wave Packet Dynamics}
\label{sec:Spherical}
This section is a brief detour into a simple special case of Gaussian wave packet dynamics that assumes that the wave packet is ``spherical'', i.e., $\mathcal{A} = a I_{d}$ and $\mathcal{B} = b I_{d}$ with $I_{d}$ being the identity matrix of size $d$; hence we replace the Siegel upper half space $\Sigma_{d}$ by $\Sigma_{1}$ even if $d \neq 1$.
We also introduce the Darboux coordinates for the resulting semiclassical dynamics; they will be later exploited in the harmonic oscillator example in Section~\ref{sec:SemiclassicalHO} to find the action--angle coordinates.

\subsection{Spherical Gaussian Wave Packet Dynamics}
Setting $\mathcal{A} = a I_{d}$ and $\mathcal{B} = b I_{d}$ in Eq.~\eqref{eq:chi} gives the ``spherical'' Gaussian wave packet, i.e., 
\begin{equation*}
  \chi(y;x) = \exp\braces{ \frac{i}{\hbar}\brackets{ \frac{1}{2}(a + ib)|x - q|^{2} + p \cdot (x - q) + (\phi + i \delta) } }.
\end{equation*}
The manifold $\mathcal{M}$ is now
\begin{equation*}
 \mathcal{M} = T^{*}\R^{d} \times \Sigma_{1} \times \mathbb{S}^{1} \times \R.
\end{equation*}
Note that the Siegel upper half space $\Sigma_{1} \cong \setdef{ a + ib \in \C }{ b > 0 }$ is literally the upper half space of $\C$.
The manifold $\mathcal{M}$ is $(2d + 4)$-dimensional, and is parametrized by 
\begin{equation*}
  y \defeq (q, p, a, b, \phi, \delta).
\end{equation*}
The symplectic one-form $\Theta_{\mathcal{M}}$, Eq.~\eqref{eq:Theta_M}, now becomes
\begin{equation*}
  \Theta_{\mathcal{M}} \defeq \iota^{*}\Theta = \mathcal{N}(b,\delta)\parentheses{
    p_{i}\,{\bf d}q^{i}
    - \frac{d\hbar}{4b}\,{\bf d}a
    - {\bf d}\phi
  }
\end{equation*}
with
\begin{equation*}
  \mathcal{N}(b,\delta) \defeq \parentheses{ \frac{\pi\hbar}{b} }^{d/2}\, \exp\parentheses{ -\frac{2\delta}{\hbar} },
\end{equation*}
and hence the symplectic form $\Omega_{\mathcal{M}}$ on $\mathcal{M}$ is
\begin{multline*}
  \Omega_{\mathcal{M}} = \mathcal{N}(b,\delta)\biggl[
  {\bf d}q^{i} \wedge {\bf d}p_{i} - \frac{d\,p_{i}}{2b} {\bf d}q^{i} \wedge {\bf d}b
  - \frac{2p_{i}}{\hbar}{\bf d}q^{i} \wedge {\bf d}\delta
  \\
  + \frac{d(d+2)\hbar}{8b^{2}} {\bf d}a \wedge {\bf d}b
  + \frac{d}{2b} \parentheses{
    {\bf d}a \wedge {\bf d}\delta - {\bf d}b \wedge {\bf d}\phi
  }
  + \frac{2}{\hbar} {\bf d}\phi \wedge {\bf d}\delta
  \biggr],
\end{multline*}
which is given by \citet{FaLu2006} (see also \citet[Section~II.4]{Lu2008}).

On the other hand, the Hamiltonian $H: \mathcal{M} \to \R$, Eq.~\eqref{eq:H-def}, is given by
\begin{align}
  H &= \mathcal{N}(b,\delta) \brackets{ \frac{1}{2m} \parentheses{ p^{2} + d\hbar\,\frac{a^{2} + b^{2}}{2b} } } + \exval{V}(q, b, \delta)
  \nonumber\\
  &= \mathcal{N}(b,\delta) \brackets{ \frac{1}{2m} \parentheses{ p^{2} + d\hbar\,\frac{a^{2} + b^{2}}{2b} } + \overline{\exval{V}}(q, b) },
  \label{eq:H-spherical}
\end{align}
where
\begin{equation*}
  \exval{V}(q, b, \delta) \defeq \ip{\chi}{V\chi}
  = \exp\parentheses{ -\frac{2\delta}{\hbar} } \int_{\R^{d}} V(x) \exp\parentheses{ -\frac{b}{\hbar}|x - q|^{2} } dx,
\end{equation*}
and
\begin{equation}
  \label{eq:avgV-spherical}
  \overline{\exval{V}}(q, b)
  \defeq \frac{\exval{V}(q, b, \delta)}{\mathcal{N}(b,\delta)}
  = \parentheses{\frac{b}{\pi\hbar}}^{d/2} \int_{\R^{d}} V(x) \exp\parentheses{ -\frac{b}{\hbar}|x - q|^{2} } dx.
\end{equation}
Hence, as shown in \cite{FaLu2006}, the Hamiltonian system~\eqref{eq:HamiltonianSystem-X_H}, i.e.,
\begin{equation*}
  {\bf i}_{X_{H}} \Omega_{\mathcal{M}} = {\bf d}H
\end{equation*}
with
\begin{equation*}
  X_{H} = \dot{q}^{i} \pd{}{q^{i}} + \dot{p}_{i} \pd{}{p_{i}} + \dot{a} \pd{}{a} + \dot{b} \pd{}{b} + \dot{\phi} \pd{}{\phi} + \dot{\delta} \pd{}{\delta}
\end{equation*}
gives the spherical version of the equations of \citet{He1975a}:
\begin{equation}
  \label{eq:Heller-spherical}
  \begin{array}{c}
    \DS
    \dot{q} = \frac{p}{m},
    \qquad
    \dot{p} = -\overline{\exval{ \nabla{V} }},
    \qquad
    \dot{a} = -\frac{a^{2} - b^{2}}{m} - \frac{1}{d} \overline{\exval{ \Delta V }},
    \qquad
    \dot{b} = -\frac{2 a b}{m},
    \medskip\\
    \DS
    \dot{\phi} = \frac{p^{2}}{2m} - \overline{\exval{V}} - \frac{d\hbar}{2m}\,b + \frac{\hbar}{4b} \overline{\exval{ \Delta V }},
    \qquad
    \dot{\delta} = \frac{d \hbar}{2m}\,a.
  \end{array}
\end{equation}

We may apply the symplectic reduction in Theorem~\ref{thm:reduction} to obtain the following reduced symplectic form on $\overline{\mathcal{M}}_{\hbar}$:
\begin{equation*}
  \overline{\Omega}_{\hbar} = {\bf d}q^{i} \wedge {\bf d}p_{i} + \frac{d\hbar}{4b^{2}} {\bf d}a \wedge {\bf d}b.
\end{equation*}
The reduced Hamiltonian~\eqref{eq:H-reduced} is now
\begin{equation*}
  \overline{H} = \frac{p^{2}}{2m} + d\hbar\,\frac{a^{2} + b^{2}}{4m\,b} + \overline{\exval{V}}(q, b),
\end{equation*}
and the reduced equations~\eqref{eq:Heller-reduced} become
\begin{equation}
  \label{eq:Heller-spherical}
  \begin{array}{c}
    \DS
    \dot{q} = \frac{p}{m},
    \qquad
    \dot{p} = -\overline{\exval{ \nabla{V} }},
    \qquad
    \dot{a} = -\frac{a^{2} - b^{2}}{m} - \frac{1}{d} \overline{\exval{ \Delta V }},
    \qquad
    \dot{b} = -\frac{2 a b}{m}.
  \end{array}
\end{equation}

\subsection{Darboux Coordinates}
\label{ssec:DarbouxCoords}
Let us define the new coordinate system
\begin{equation}
  \label{eq:DarbouxCoords}
  (q^{i}, r, \varphi, \tilde{p}_{i}, p_{r}, p_{\varphi}) \defeq \parentheses{ q^{i},\, \frac{\sqrt{d}\,\hbar}{2}\frac{\mathcal{N}(b,\delta)}{b},\, -\phi,\, \mathcal{N}(b,\delta)\,p,\, \frac{\sqrt{d}}{2}\,a,\, \mathcal{N}(b,\delta) }.
\end{equation}
Then, the symplectic form $\Omega_{\mathcal{M}}$ takes the canonical form
\begin{equation*}
  \Omega_{\mathcal{M}} = {\bf d}q^{i} \wedge {\bf d}\tilde{p}_{i} + {\bf d}r \wedge {\bf d}p_{r} + {\bf d}\varphi \wedge {\bf d}p_{\varphi},
\end{equation*}
and thus the above coordinates are the Darboux coordinates.
Hence, the Hamiltonian system~\eqref{eq:Heller-spherical} is transformed to the following canonical form:
\begin{equation*}
  \begin{array}{c}
    \DS
    \dot{q}^{i} = \pd{H}{p_{i}},
    \qquad
    \dot{r} = \pd{H}{p_{r}},
    \qquad
    \dot{\varphi} = \pd{H}{p_{\varphi}},
    \medskip\\
    \DS
    \dot{\tilde{p}}_{i} = -\pd{H}{q^{i}},
    \qquad
    \dot{p}_{r} = -\pd{H}{r},
    \qquad
    \dot{p}_{\varphi} = -\pd{H}{\varphi}.
  \end{array}
\end{equation*}

The Darboux coordinates~\eqref{eq:DarbouxCoords} for $\mathcal{M}$ induce those for $\overline{\mathcal{M}}_{\hbar}$ as follows:
On the level set ${\bf J}_{\!\mathcal{M}}^{-1}(-\hbar)$, we have $\mathcal{N}(b,\delta) = 1$ and thus $\tilde{p} = p$; therefore we have the Darboux coordinates $(q^{i}, p_{i}, r, p_{r})$ for $\overline{\mathcal{M}}_{\hbar}$ with
\begin{equation*}
  (r, p_{r}) = \frac{\sqrt{d}}{2} \parentheses{ \frac{\hbar}{b},\, a },
\end{equation*}
that is, the reduced symplectic form $\overline{\Omega}_{\hbar}$ takes the canonical form
\begin{equation}
  \label{eq:omega-canonical}
  \overline{\Omega}_{\hbar} = {\bf d}q^{i} \wedge {\bf d}p_{i} + {\bf d}r \wedge {\bf d}p_{r}.
\end{equation}
Therefore, the reduced dynamics~\eqref{eq:HamiltonianSystem-X_h} with
\begin{equation*}
  X_{\overline{H}} = \dot{q}^{i} \pd{}{q^{i}} + \dot{p}_{i} \pd{}{p_{i}} + \dot{r} \pd{}{r} + \dot{p}_{r} \pd{}{p_{r}}
\end{equation*}
is written as canonical Hamilton's equations:
\begin{equation*}
  \dot{q}^{i} = \pd{\overline{H}}{p_{i}},
  \qquad
  \dot{r} = \pd{\overline{H}}{p_{r}},
  \qquad
  \dot{p}_{i} = -\pd{\overline{H}}{q^{i}},
  \qquad
  \dot{p}_{r} = -\pd{\overline{H}}{r}.
\end{equation*}

\begin{remark}
  \label{rem:DarbouxCoords}
  That the variables $b^{-1}$ and $a$ are essentially canonically conjugate was pointed out by \citet{Li1988} and \citet{SiSuMu1988}.
  See also \citet{BrLaVa1989} and \citet{PaSc1994}.
\end{remark}

\section{Reconstruction---Dynamic and Geometric Phases}
\label{sec:ReconstructionAndGeometricPhase}
\subsection{Theory of Reconstruction}
As described in Section~\ref{ssec:GeometryOfGWPDyn}, the Gaussian wave packet dynamics $X_{H}$ defined by \eqref{eq:HamiltonianSystem-X_H} in $\mathcal{M}$ may be reduced to the Hamiltonian dynamics $X_{\overline{H}}$ defined by \eqref{eq:HamiltonianSystem-X_h} in the reduced symplectic manifold $\overline{\mathcal{M}}_{\hbar} \defeq {\bf J}_{\mathcal{M}}^{-1}(-\hbar)/\mathbb{S}^{1}$.
Now let $\bar{c}(t)$ be an integral curve of the reduced dynamics $X_{\overline{H}}$, i.e., $\dot{\bar{c}}(t) = X_{\overline{H}}(\bar{c}(t))$.
Then, the curve $\bar{c}(t)$ is the projection of an integral curve $c(t)$ of the full dynamics $X_{H}$ on ${\bf J}_{\mathcal{M}}^{-1}(-\hbar)$, i.e., $\pi_{\hbar} \circ c(t) = \bar{c}(t)$.
Then, a natural question to ask is: {\em Given the reduced dynamics $\bar{c}(t)$, is it possible to construct the full dynamics $c(t)$?}
The theory of {\em reconstruction}~(\citet{MaMoRa1990}; see also~\citet[Chapter~6]{Ma1992}) provides an answer to the question, and the so-called dynamic and geometric phases arise naturally when reconstructing the full dynamics from the geometric point of view.

\subsection{Dynamic Phase}
For the full quantum dynamics with the Schr\"odinger equation, we may define a principal connection form $\mathscr{A}: T{\bf J}^{-1}(-\hbar) \to \mathfrak{so}(2)$ on the principal bundle ${\bf J}^{-1}(-\hbar) \to \mathbb{P}(\mathcal{H})$ as follows~(see \citet{Si1983} and \citet[Section~13.1]{Mo2002}):
\begin{equation*}
  \mathscr{A}(\psi) = \Im\ip{\psi}{{\bf d}\psi}|_{T_{\psi}{\bf J}^{-1}(-\hbar)};
  \quad
  \ip{\mathscr{A}(\psi)}{v_{\psi}} = \Im\ip{\psi}{v_{\psi}} \ \text{for}\ v_{\psi} \in T_{\psi}{\bf J}^{-1}(-\hbar),
\end{equation*}
that is, $\mathscr{A}$ is $-\Theta/\hbar$ restricted to ${\bf J}^{-1}(-\hbar)$.
Since $\norm{\psi}^{2} = \ip{\psi}{\psi} = 1$ for $\psi \in {\bf J}^{-1}(-\hbar)$, we have $\ip{{\bf d}\psi}{\psi} + \ip{\psi}{{\bf d}\psi} = 0$, and thus
\begin{equation*}
  \mathscr{A}(\psi) = -i\ip{\psi}{{\bf d}\psi}.
\end{equation*}
This induces the principal connection form $\mathscr{A}_{\mathcal{M}}: T{\bf J}_{\mathcal{M}}^{-1}(-\hbar) \to \mathfrak{so}(2)$ on the principal bundle ${\bf J}_{\mathcal{M}}^{-1}(-\hbar) \to \overline{\mathcal{M}}_{\hbar}$ that is given as follows:
\begin{equation}
  \label{eq:mathcalA_M}
  \mathscr{A}_{\mathcal{M}} \defeq \iota^{*} \mathscr{A}
  = \frac{1}{\hbar}{\bf d}\phi
  -\frac{1}{\hbar}\,p_{i}\,{\bf d}q^{i}
  + \frac{1}{4}\tr(\mathcal{B}^{-1}{\bf d}\mathcal{A}),
\end{equation}
which, for the spherical case, reduces to 
\begin{equation}
  \label{eq:mathcalA_M-spherical}
  \mathscr{A}_{\mathcal{M}}
  = \frac{1}{\hbar}{\bf d}\phi
  -\frac{1}{\hbar}\,p_{i}\,{\bf d}q^{i}
  + \frac{d}{4b}\,{\bf d}a
  = \frac{1}{\hbar}\parentheses{ -{\bf d}\varphi
    - p_{i}\,{\bf d}q^{i}
    + r\,{\bf d}p_{r}
  }.
\end{equation}

Now, let $y_{0}$ be a point in ${\bf J}_{\mathcal{M}}^{-1}(-\hbar)$ and $d(t)$ be the {\em horizontal lift} of the curve $\bar{c}(t)$ such that $d(0) = y_{0}$, i.e., the curve defined uniquely by $\pi_{\hbar} \circ d(t) = \bar{c}(t)$ and $d(0) = y_{0}$ with
\begin{equation}
  \label{eq:horizontal_lift}
  \mathscr{A}_{\mathcal{M}}(d(t)) \cdot \dot{d}(t) = 0.
\end{equation}
Then, since the full dynamics $c(t)$ satisfies $\pi_{\hbar} \circ c(t) = \bar{c}(t)$, we have $\pi_{\hbar} \circ c(t) = \pi_{\hbar} \circ d(t)$, and thus there exists a curve $g(t)$ in $\mathbb{S}^{1}$ such that $c(t) = g(t)\,d(t)$.
By the Reconstruction Theorem (\cite[Section~2A]{MaMoRa1990} and \cite[Section~6.2]{Ma1992}), the curve $g(t)$ in $\mathbb{S}^{1}$ is given by
\begin{equation}
  \label{eq:g}
  g(t) = \exp\parentheses{ i\int_{0}^{t} \xi(s)\,ds },
\end{equation}
where
\begin{equation*}
  \xi(t) \defeq \mathscr{A}_{\mathcal{M}}(d(t)) \cdot X_{H}(d(t))
\end{equation*}
is a curve in $\mathfrak{so}(2) \cong \R$.
It is straightforward to see, from Eqs.~\eqref{eq:H}, \eqref{eq:Heller-spherical}, and \eqref{eq:mathcalA_M-spherical}, that
\begin{equation}
  \label{eq:xi}
  \xi(t) = -\frac{H(d(t))}{\hbar} = -\frac{H(c(t))}{\hbar} = -\frac{E}{\hbar},
\end{equation}
where the the second equality follows from the $\mathbb{S}^{1}$-invariance of the Hamiltonian $H$; the last equality follows since $c(t)$ is an integral curve of $X_{H}$, and so the Hamiltonian $H$ is constant along $c(t)$, and its value is determined by the initial condition $E \defeq H(c(0))$.
Therefore, we obtain
\begin{equation*}
  g(t) = \exp\parentheses{ -\frac{i}{\hbar}E\,t },
\end{equation*}
which is compatible with the result for the full quantum dynamics (see, e.g., \citet[Section~13.2]{Mo2002}).
Then, the dynamic phase $g_{\text{dyn}} \in \mathbb{S}^{1}$ achieved over the time interval $[0,T]$ is given by
\begin{equation*}
  g_{\text{dyn}} = \exp\parentheses{ \frac{i}{\hbar} \Delta\phi_{\text{dyn}} } = \exp\parentheses{ -\frac{i}{\hbar}E\,T }.
\end{equation*}
where $\Delta\phi_{\text{dyn}}$ is the change in the angle variable $\phi$ in the coordinates for $\mathcal{M}$ (see also Eq.~\eqref{eq:chi}):
\begin{equation*}
  \Delta\phi_{\text{dyn}} = -E\,T.
\end{equation*}

\subsection{Geometric Phase}
The curvature of the principal connection form~\eqref{eq:mathcalA_M} is given by
\begin{equation*}
  \mathscr{B}_{\mathcal{M}} = {\bf d}\mathscr{A}_{\mathcal{M}}
  = \frac{1}{\hbar}\parentheses{
    {\bf d}q^{i} \wedge {\bf d}p_{i} + \frac{\hbar}{4} \mathcal{B}^{-1}_{ik} \mathcal{B}^{-1}_{lj} \d\mathcal{A}_{ij} \wedge \d\mathcal{B}_{kl}
  },
\end{equation*}
and, for the spherical case, we have
\begin{equation*}
  \mathscr{B}_{\mathcal{M}}
  = \frac{1}{\hbar}\parentheses{
    {\bf d}q^{i} \wedge {\bf d}p_{i} + \frac{d\hbar}{4b^{2}} {\bf d}a \wedge {\bf d}b
  }.
\end{equation*}
Therefore, its reduced curvature form, i.e., $\mathscr{B}_{\mathcal{M}}$ viewed as a two-form on $\overline{\mathcal{M}}_{\hbar}$, becomes
\begin{equation*}
  \overline{\mathscr{B}}_{\hbar} = \frac{1}{\hbar}\,\overline{\Omega}_{\hbar}.
\end{equation*}

Suppose that the curve of the reduced dynamics on $\overline{\mathcal{M}}_{\hbar}$ is closed with period $T$, i.e., $\bar{c}(0) = \bar{c}(T)$ for some $T > 0$.
Then, the geometric phase (holonomy) $g_{\text{geom}} \in \mathbb{S}^{1}$ achieved over the period $T$ is defined by
\begin{equation*}
  d(T) = g_{\text{geom}}\,d(0).
\end{equation*}
Let $D$ be any two-dimensional submanifold of $\overline{\mathcal{M}}_{\hbar}$ whose boundary is the curve $\bar{c}([0,T))$; then the geometric phase is given by the following reconstruction phase (see, e.g., \citet[Corollary~4.2]{MaMoRa1990}):
\begin{equation*}
  g_{\text{geom}} = \exp\parentheses{ \frac{i}{\hbar} \Delta\phi_{\text{geom}} } = \exp\parentheses{ -i\iint_{D} \bar{\mathscr{B}}_{\hbar} }
  = \exp\parentheses{ -\frac{i}{\hbar} \iint_{D} \overline{\Omega}_{\hbar} } \in \mathbb{S}^{1}.
\end{equation*}
where $\Delta\phi_{\text{geom}}$ is the change in the angle variable $\phi$:
\begin{equation}
  \label{eq:GeometricPhase}
  \Delta\phi_{\text{geom}} = -\iint_{D} \overline{\Omega}_{\hbar}.
\end{equation}
This generalizes the result of \citet{An1988b,An1990,An1991}, which was derived for the frozen Gaussian wave packet, i.e., the spherical case with $a$ and $b$ being constant.

Notice that {\em we derived the above formula as a reconstruction of the Hamiltonian dynamics in $\mathcal{M}$}; namely, we have incorporated the phase variable $\phi$ (accompanied by $\delta$) into the expression of the Gaussian wave packet \eqref{eq:chi} to write the full dynamics in $\mathcal{M}$ as a Hamiltonian system (see the discussion just above Remark~\ref{rem:normalization}), and the reconstruction of the dynamics on $\mathcal{M}$ from the reduced dynamics on $\overline{\mathcal{M}}_{\hbar}$ gave rise to the geometric phase.
This gives a natural geometric account (and generalization) of the somewhat ad-hoc calculations performed in \cite{An1988b,An1990,An1991}.

\subsection{Total Phase}
Combining the dynamic and geometric phases, we obtain the total phase change over the period $T$:
\begin{equation*}
  g_{\text{total}} = \exp\parentheses{ \frac{i}{\hbar} \Delta\phi_{\text{total}} } = g_{\text{dyn}} \cdot g_{\text{geom}}
  = \exp\brackets{ -\frac{i}{\hbar}\parentheses{ E\,T + \iint_{D} \overline{\Omega}_{\hbar} } },
\end{equation*}
or
\begin{equation*}
  \Delta\phi_{\text{total}} = \Delta\phi_{\text{dyn}} + \Delta\phi_{\text{geom}} =  -E\,T - \iint_{D} \overline{\Omega}_{\hbar},
\end{equation*}
which is similar to the rigid body phase of \citet{Mo1991b} (see also \citet{Ha1985}, \citet{An1988a}, and \citet{Le1993}).
Noting that the phase factor in \eqref{eq:chi} is $e^{i\phi/\hbar}$, it is convenient to rewrite the result as
\begin{equation}
  \label{eq:TotalPhase}
  \Delta \parentheses{ \frac{ \phi_{\text{total}} }{\hbar} } = \frac{1}{\hbar} \parentheses{ -E\,T - \iint_{D} \overline{\Omega}_{\hbar} }.
\end{equation}

Note that we made an assumption that the reduced dynamics on $\overline{\mathcal{M}}_{\hbar}$ defined by $X_{\overline{H}}$ is periodic with period $T$.
In Section~\ref{ssec:CalcOfGeomPhase} below, we will show that such a periodic orbit in $\overline{\mathcal{M}}_{\hbar}$ in fact exists for the semiclassical harmonic oscillator and calculate the explicit expression for the total phase.

If the reduced dynamics is not periodic, we do not have a simple formula for the phase change as above.
However, one may still obtain an expression for the phase factor $\phi$ in terms of the reduced solution $\bar{c}(t) = ( q(t), p(t), \mathcal{A}(t), \mathcal{B}(t) )$ defined by Eq.~\eqref{eq:Heller-reduced}.
Let us write
\begin{equation*}
  d(t) = (q(t), p(t), \mathcal{A}(t), \mathcal{B}(t), \vartheta(t), \delta(t)),
  \qquad
  c(t) = (q(t), p(t), \mathcal{A}(t), \mathcal{B}(t), \phi(t), \delta(t)).
\end{equation*}
Since $c(t) = g(t)\,d(t)$ with $g(t)$ given by Eq.~\eqref{eq:g},
\begin{equation*}
  \phi(t) = \hbar\int_{0}^{t} \xi(s)\,ds + \vartheta(t),
\end{equation*}
and so, using the expression for $\xi(t)$ in Eq.~\eqref{eq:xi},
\begin{equation*}
  \dot{\phi}(t) = \hbar\,\xi(t) + \dot{\vartheta}(t) = -H(d(t)) + \dot{\vartheta}(t).
\end{equation*}
Now, the horizontal lift equation~\eqref{eq:horizontal_lift} gives
\begin{align*}
  \dot{\vartheta} &= p_{i}\,\dot{q}^{i} - \frac{\hbar}{4}\tr(\mathcal{B}^{-1}\dot{\mathcal{A}})
  \\
  &= \frac{p^{2}}{m} + \frac{\hbar}{4m}\tr\brackets{ \mathcal{B}^{-1}(\mathcal{A}^{2} - \mathcal{B}^{2}) } + \frac{\hbar}{4} \tr\parentheses{ \mathcal{B}^{-1} \overline{ \exval{\nabla^{2}V} } },
\end{align*}
where we used the reduced equations~\eqref{eq:Heller-reduced}.
As a result, by using the expression for the Hamiltonian~\eqref{eq:H} and noting that $\mathcal{N}(\mathcal{B},\delta) = 1$ here, we obtain,
\begin{equation*}
  \dot{\phi} = \frac{p^{2}}{2m} - \overline{ \exval{V} }
      - \frac{\hbar}{2m} \tr\mathcal{B}
      + \frac{\hbar}{4} \tr\parentheses{ \mathcal{B}^{-1} \overline{ \exval{\nabla^{2}V} } },
\end{equation*}
thereby recovering the equation for $\phi$ in the full dynamics~\eqref{eq:Heller}.

\section{Asymptotic Evaluation of the Potential}
\label{sec:Potential}
One obstacle in practical applications of the semiclassical Hamiltonian system, Eq.~\eqref{eq:Heller} or \eqref{eq:Heller-reduced}, is the evaluation of the potential terms $\overline{\exval{V}}$, $\overline{\exval{\nabla V}}$, and $\overline{\exval{\Delta V}}$, which are generally given by complicated integrals (see Eqs.~\eqref{eq:avgV} and \eqref{eq:avgV-spherical}; originally due to \citet{CoKa1990}).
If the potential $V(x)$ is given as a simple polynomial, one may reduce the integrals to Gaussian integrals and obtain closed forms of them exactly; this is particularly easy for the spherical case (see Eq.~\eqref{eq:avgV-spherical}).
However, one rarely has such a simple potential $V(x)$ in problems of interest in chemical physics, and thus there is a need to approximate the potential terms.

As mentioned in Remark~\ref{rem:avgV}, Heller's formulation does not involve these averaged potential terms, but from our perspective, it can be interpreted as adopting the following simple approximations of the expectation values:
\begin{equation*}
  \overline{\exval{V}}(q, \mathcal{B}) \simeq V(q),
  \qquad
  \overline{\exval{\nabla V}}(q, \mathcal{B}) \simeq \nabla V(q),
  \qquad
  \overline{\exval{\Delta V}}(q, \mathcal{B}) \simeq \Delta V(q).
\end{equation*}
Notice, however, that {\em this approximation neglects non-classical effects coming from $\mathcal{B}$ altogether, and seems to be too crude for a semiclassical model}.

Instead, we apply Laplace's method to the integral in the potential term $\overline{\exval{V}}$ to obtain an asymptotic expansion of it.
As we shall see later, this also results in an asymptotic expansion of the Hamiltonian $H$, Eq.~\eqref{eq:H}, and then our Hamiltonian/symplectic viewpoint provides a correction term to the formulation by \citet{He1975a} and \citet{LeHe1982}.
The main result in this section, Proposition~\ref{prop:asymptoticV}, gives a multi-dimensional generalization of the expansion for the one-dimensional case in \citet{PaSc1994} with a rigorous justification.

\subsection{Non-spherical Case}
The key observation here is that the potential term $\overline{\exval{V}}$, Eq.~\eqref{eq:avgV} or \eqref{eq:avgV-spherical}, is given as a typical integral to which one applies Laplace's method for asymptotic evaluation of integrals, i.e., we have
\begin{equation*}
  \overline{\exval{V}}(q, \mathcal{B})
  = \sqrt{ \frac{\det \mathcal{B}}{(\pi\hbar)^{d}} }\, F_{\hbar}(q, \mathcal{B}),
\end{equation*}
where
\begin{equation}
  \label{eq:F_hbar}
  F_{\hbar}(q, \mathcal{B}) \defeq \int_{\R^{d}} e^{R(x)/\hbar}\, V(x)\,dx
\end{equation}
with
\begin{equation*}
  R(x) = -(x - q)^{T}\mathcal{B}(x - q).
\end{equation*}
Now, an asymptotic evaluation of the integral $F_{\hbar}(q, \mathcal{B})$ gives us the following:
\begin{proposition}
  \label{prop:asymptoticV}
  If the potential $V(x)$ is a smooth function such that $e^{\sigma R(x)/\hbar} V(x)$ is square integrable in $\R^{d}$ for some $\sigma \in [0,1)$, then the potential term $\overline{\exval{V}}$ has the asymptotic expansion
  \begin{subequations}
    \label{eq:asymptoticV}
    \begin{equation}
      \overline{\exval{V}}(q, \mathcal{B}) \sim \sum_{n=0}^{\infty} c_{n}(q, \mathcal{B})\, \hbar^{n}
      \quad \text{as} \quad \hbar \to 0,
    \end{equation}
    where
    \begin{equation}
      c_{n}(q, \mathcal{B}) \defeq \frac{1}{4^{n}} \sum_{\substack{j_{1} + \dots + j_{d} = 2n\\\text{$j_{k}$ all even}}}
      \frac{g_{j}(q)}{\prod_{k=1}^{d} b_{k}^{j_{k}/2} (j_{k}/2)!}
    \end{equation}
    and
    \begin{equation}
      g_{j}(\xi) \defeq D^{j}\tilde{V}(\mathcal{Q} \xi)
      = \pd{^{|j|}}{\xi_{1}^{j_{1}} \partial \xi_{2}^{j_{2}} \dots \partial \xi_{d}^{j_{d}}}\, \tilde{V}(\mathcal{Q} \xi),
    \end{equation}
    with $\tilde{V}(\xi) \defeq V(q + \xi)$; $b_{1}, \dots b_{d}$ are the eigenvalues of $\mathcal{B}$, and $\mathcal{Q}$ is the orthogonal matrix such that
    \begin{equation*}
      \mathcal{B} \mathcal{Q} = \mathcal{Q} \diag(b_{1}, \dots, b_{d}),
    \end{equation*}
    i.e., each of its columns is an eigenvector of $\mathcal{B}$.
  \end{subequations}
\end{proposition}

\begin{proof}
  The asymptotic expansion follows from a standard result of Laplace's method (see, e.g., \citet[Section~3.7]{Mi2006}) applied to the integral $F_{\hbar}(q, \mathcal{B})$ in Eq.~\eqref{eq:F_hbar} restricted to a neighborhood of the point $x = q$.
  Hence, we need an estimate of the contribution from the remaining part of the integral to justify the expansion.
  See Appendix~\ref{sec:asymptoticV-proof} for this estimate.
\end{proof}

In particular, we can rewrite the first two terms more explicitly:
\begin{equation}
  \label{eq:avgV-asymptotic}
  \overline{\exval{V}}(q, \mathcal{B}) = V(q) + \frac{\hbar}{4} \tr\brackets{ \mathcal{B}^{-1} \nabla^{2}V(q) } + O(\hbar^{2})
  \quad \text{as} \quad \hbar \to 0.
\end{equation}
Therefore, the Hamiltonian~\eqref{eq:H} becomes, as $\hbar \to 0$,
\begin{equation*}
  H = \mathcal{N}(\mathcal{B},\delta) \braces{
    \frac{p^{2}}{2m} + V(q)
    + \frac{\hbar}{4m}\tr\brackets{ \mathcal{B}^{-1}(\mathcal{A}^{2} + \mathcal{B}^{2}) }
    + \frac{\hbar}{4} \tr\parentheses{ \mathcal{B}^{-1} \nabla^{2}V(q) }
    + O(\hbar^{2})
  }.
\end{equation*}
We may then neglect the second-order term $O(\hbar^{2})$ to obtain an approximate Hamiltonian
\begin{equation*}
  H \simeq H_{1}
  \defeq \mathcal{N}(\mathcal{B},\delta) \braces{
    \frac{p^{2}}{2m} + V(q)
    + \frac{\hbar}{4}\tr\brackets{ \mathcal{B}^{-1}\parentheses{ \frac{\mathcal{A}^{2} + \mathcal{B}^{2}}{m} + \nabla^{2}V(q) } }
  }.
\end{equation*}
Then, the Hamiltonian system ${\bf i}_{X_{H_{1}}} \Omega_{\mathcal{M}} = {\bf d}H_{1}$ gives the following approximation to Eq.~\eqref{eq:Heller}:
\begin{equation}
  \label{eq:Heller-asymptotic}
  \begin{array}{c}
    \DS
    \dot{q} = \frac{p}{m},
    \qquad
    \dot{p} = -\pd{}{q}\brackets{ V(q) + \frac{\hbar}{4} \tr\parentheses{ \mathcal{B}^{-1} \nabla^{2}V(q) } },
    \medskip\\
    \DS
    \dot{\mathcal{A}} = -\frac{1}{m}(\mathcal{A}^{2} - \mathcal{B}^{2}) - \nabla^{2}V(q),
    \qquad
    \dot{\mathcal{B}} = -\frac{1}{m}(\mathcal{A}\mathcal{B} + \mathcal{B}\mathcal{A}),
    \medskip\\
    \DS
    \dot{\phi} = \frac{p^{2}}{2m} - V(q) - \frac{\hbar}{2m} \tr\mathcal{B},
    \qquad
    \dot{\delta} = \frac{\hbar}{2m} \tr\mathcal{A}.
  \end{array}
\end{equation}
Notice a slight difference from those equations obtained in \citet{He1975a} and \citet{LeHe1982}: {\em The second equation above has a semiclassical correction proportional to $\hbar$, whereas those in \cite{He1975a, LeHe1982} are missing this term.}
Furthermore, since the correction term generally depends on $\mathcal{B}$, {\em the equations for $q$ and $p$ are not decoupled as in \citet{He1975a}.
Therefore, it is crucial to formulate the whole system---as opposed to those for $q$ and $p$ only---as a Hamiltonian system}.
We will see in Section~\ref{sec:SemiclassicalTunneling} that this semiclassical correction term in fact realizes a classically forbidden motion.

\begin{remark}
  \label{rem:quadraticV}
  If the potential $V(x)$ is quadratic, then the asymptotic expansion~\eqref{eq:asymptoticV} terminates at the second term, i.e., $c_{n} = 0$ for $n \ge 2$, and becomes exact.
  Hence $H = H_{1}$ and so Eqs.~\eqref{eq:Heller} and \eqref{eq:Heller-asymptotic} are equivalent.
  Moreover, since $\nabla^{2}V(q)$ is now constant, the second equation in \eqref{eq:Heller-asymptotic} reduces to the canonical one, and hence Eq.~\eqref{eq:Heller-asymptotic} reduces to those of \citet{He1975a} and \citet{LeHe1982}.
\end{remark}

One may reduce Eq.~\eqref{eq:Heller-asymptotic} just as in Theorem~\ref{thm:reduction}: The reduced system $\i_{X_{\overline{H}_{1}}} \overline{\Omega}_{\hbar} = \d{\overline{H}_{1}}$ with the reduced approximate Hamiltonian
\begin{equation*}
  \overline{H}_{1} \defeq \frac{p^{2}}{2m} + V(q)
    + \frac{\hbar}{4}\tr\brackets{ \mathcal{B}^{-1}\parentheses{ \frac{\mathcal{A}^{2} + \mathcal{B}^{2}}{m} + \nabla^{2}V(q) } }
\end{equation*}
gives the first four equations of \eqref{eq:Heller-asymptotic}.
Notice that the Hamiltonian is split into the classical one and a semiclassical correction proportional to $\hbar$.

\subsection{Spherical Case}
The asymptotic expansion for the spherical model from Section~\ref{sec:Spherical} follows easily from Eq.~\eqref{eq:avgV-asymptotic}: Setting $\mathcal{B} = b I_{d}$ gives
\begin{equation}
  \label{eq:avgV-asymptotic-spherical}
  \overline{\exval{V}}(q, b) = V(q) + \frac{\hbar}{4b} \Delta V(q) + O(\hbar^{2})
  \quad \text{as} \quad \hbar \to 0.
\end{equation}
Note that higher-order terms are easy to calculate for the spherical case, because $\mathcal{B} = b I_{d}$ implies that $b_{k} = b$ for $k = 1, \dots, d$ and $\mathcal{Q} = I_{d}$.
Now, Eq.~\eqref{eq:Heller-asymptotic} becomes
\begin{equation*}
  \begin{array}{c}
    \DS
    \dot{q} = \frac{p}{m},
    \qquad
    \dot{p} = -\pd{}{q}\brackets{ V(q) + \frac{\hbar}{4b} \Delta V(q) },
    \qquad
    \dot{a} = -\frac{1}{m}(a^{2} - b^{2}) - \frac{1}{d}\,\Delta V(q),
    \qquad
    \dot{b} = -\frac{2 a b}{m},
    \medskip\\
    \DS
    \dot{\phi} = \frac{p^{2}}{2m} - V(q) - \frac{d\hbar}{2m}\,b,
    \qquad
    \dot{\delta} = \frac{d\hbar}{2m}\,a.
  \end{array}
\end{equation*}

\section{Example~1: Semiclassical Harmonic Oscillator}
\label{sec:SemiclassicalHO}
In this section, we illustrate the theory developed so far by considering a simple one-dimensional harmonic oscillator.
For this special case, the system~\eqref{eq:Heller-spherical} is easily integrable as shown by \citet{He1975a}; however, we approach the problem from a more Hamiltonian perspective.
Namely, we first find the action--angle coordinates for the reduced system using the Darboux coordinates from Section~\ref{ssec:DarbouxCoords}.
As we shall see later, the action--angle coordinates give an insight into the periodic motion of the system, and facilitates our calculation of the geometric phase.

\subsection{The Hamilton--Jacobi Equation and Separation of Variables}
Consider the one-dimensional harmonic oscillator, i.e., $d = 1$ and
\begin{equation*}
  V(x) = \frac{1}{2}m\,\omega^{2}x^{2}.
\end{equation*}
Note that for the one-dimensional case, the non-spherical wave packet reduces to the spherical one.
Then, the potential term is easily calculated to give\footnote{Since $V(x)$ is quadratic, the asymptotic expansion~\eqref{eq:asymptoticV} is exact, i.e., $c_{n} = 0$ for $n \ge 2$, and so \eqref{eq:asymptoticV} gives the same result.}
\begin{equation*}
  \overline{\exval{V}}(q, b) = V(q) + \frac{m\,\omega^{2}\hbar}{4b},
\end{equation*}
and so the Hamiltonian \eqref{eq:H-spherical} is
\begin{equation*}
  H = \mathcal{N}(b,\delta) \brackets{ \frac{1}{2m} \parentheses{ p^{2} + \hbar\,\frac{a^{2} + b^{2}}{2b} } + \frac{m\,\omega^{2}}{2}\parentheses{ q^{2} + \frac{\hbar}{2b} } },
\end{equation*}
or, using the Darboux coordinates defined in Eq.~\eqref{eq:DarbouxCoords},
\begin{equation}
  \label{eq:H-Darboux-HO}
  H = \frac{p^{2}}{2m\,p_{\varphi}}
    + p_{\varphi}\frac{m\,\omega^{2}}{2} q^{2}
    + \frac{2}{m}\,r\,p_{r}^{2}
    + \frac{m\,\omega^{2}}{2}r
    + \frac{\hbar^{2}}{8 m r}\,p_{\varphi}^{2}.
\end{equation}
Then, the reduced Hamiltonian \eqref{eq:H-reduced} becomes
\begin{align*}
  \overline{H} &= \frac{1}{2m} \parentheses{ p^{2} + \hbar\,\frac{a^{2} + b^{2}}{2b} } + \frac{m\,\omega^{2}}{2}\parentheses{ q^{2} + \frac{\hbar}{2b} }
  \\
  &= \frac{1}{2m} \parentheses{ p^{2} + 4r\,p_{r}^{2} }
    + \frac{m\,\omega^{2}}{2} (q^{2} + r)
    + \frac{\hbar^{2}}{8 m r},
\end{align*}
which also follows from Eq.~\eqref{eq:H-Darboux-HO} with $p_{\varphi} = 1$.

The Hamilton--Jacobi equation for the reduced dynamics 
\begin{equation*}
  \overline{H}\parentheses{ q, r, \pd{W}{q}, \pd{W}{r} } = E
\end{equation*}
with the ansatz $W(q, r) = W_{q}(q) + W_{r}(r)$ gives
\begin{equation*}
  \frac{1}{2m} \parentheses{ \od{W_{q}}{q} }^{2}
  + \frac{2r}{m} \parentheses{ \od{W_{r}}{r} }^{2}
  + \frac{m\,\omega^{2}}{2} (q^{2} + r)
  + \frac{\hbar^{2}}{8 m r} = E.
\end{equation*}
Hence, by separation of variables, we obtain
\begin{equation*}
  \frac{1}{2m} \parentheses{ \od{W_{q}}{q} }^{2}
  + \frac{m\,\omega^{2}}{2} q^{2} = E_{1},
  \qquad
  \frac{2r}{m} \parentheses{ \od{W_{r}}{r} }^{2}
  + \frac{m\,\omega^{2}}{2}r
  + \frac{\hbar^{2}}{8 m r} = E_{r},
\end{equation*}
where $E_{1}$ and $E_{r}$ are constants such that $E_{1} + E_{r} = E$.
Thus,
\begin{equation*}
  \od{W_{q}}{q} = \pm \sqrt{
     2 m E_{1} - m^2\omega^{2} q^{2}
  },
  \qquad
  \od{W_{r}}{r} = \pm \frac{m\,\omega}{2} \sqrt{
    -1 + \frac{\alpha}{r} - \frac{L^{2}}{2r^{2}}
  },
\end{equation*}
where
\begin{equation*}
  \alpha \defeq \frac{2E_{r}}{m\,\omega^{2}},
  \qquad
  L \defeq \frac{\hbar}{\sqrt{2}\,m\,\omega},
\end{equation*}
and we assume that $L < \alpha/\sqrt{2}$ is satisfied.

\subsection{Action--Angle Coordinates}
The above solution of the Hamilton--Jacobi equation gives rise to the canonical coordinate transformation to the action--angle coordinates, i.e., $(q, r, p, p_{r}) \mapsto (\theta_{1}, \theta_{r}, I_{1}, I_{r})$.

The first pair of action--angle coordinates $(\theta_{1}, I_{1})$ are those for the classical harmonic oscillator:
Let $\gamma_{1}$ be the curve (clockwise orientation) on the $q$-$p$ plane defined by $p^{2} = \parentheses{ \tod{W_{q}}{q} }^{2}$, i.e.,
\begin{equation*}
  \frac{1}{2m} p^{2} + \frac{m\,\omega^{2}}{2} q^{2} = E_{1},
\end{equation*}
which is an ellipse whose semi-major and semi-minor axes are $\sqrt{2E_{1}/m}/\omega$ and $\sqrt{2 m E_{1}}$.
Therefore, the first action variable $I_{1}$ is given by Stokes' theorem as follows:
\begin{equation*}
  I_{1} = \frac{1}{2\pi} \oint_{\gamma_{1}} p\,{\bf d}q
  = \frac{1}{2\pi} \int_{A_{1}} {\bf d}p \wedge {\bf d}q
  = \frac{E_{1}}{\omega}
  = \frac{1}{\omega} \parentheses{ \frac{1}{2m} p^{2} + \frac{m\,\omega^{2}}{2} q^{2} },
\end{equation*}
where $A_{1}$ is the area inside the ellipse (with the orientation compatible with that of $\gamma_{1}$; see Fig~\ref{fig:PeriodicOrbits}), i.e., $\partial A_{1} = \gamma_{1}$; hence the surface integral is the area of the ellipse.

\begin{figure}[htbp]
  \centering
  \subfigure{
    \includegraphics[width=.4\linewidth]{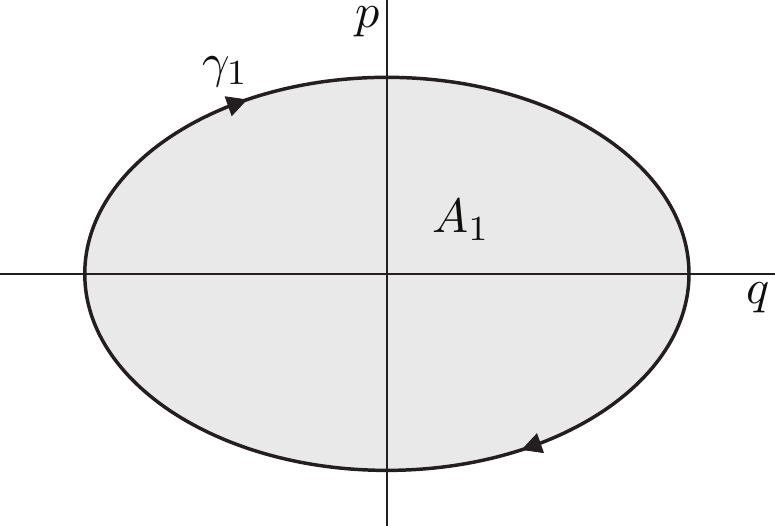}
  }
  \qquad
  \subfigure{
    \includegraphics[width=.4\linewidth]{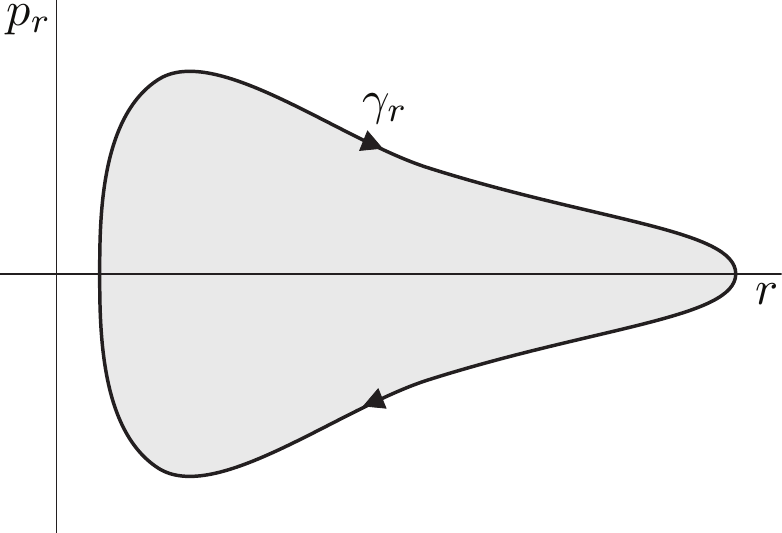}
  }
  \caption{Periodic orbits on the $q$-$p$ and $r$-$p_{r}$ planes.}
  \label{fig:PeriodicOrbits}
\end{figure}

The angle variable $\theta_{1}$ is then
\begin{equation*}
  \theta_{1} = \pd{}{I_{1}} \int \od{W_{q}}{q}\,dq
  = \int \sqrt{
     2 m\,\omega I_{1} - m^2\omega^{2} q^{2}
  }\,dq = \tan^{-1}\parentheses{ \frac{m\,\omega\, q}{p} }
\end{equation*}

Interestingly, the second pair of action--angle coordinates $(\theta_{r}, I_{r})$ is essentially the same as those for the radial part of the planar Kepler problem (see, e.g., \citet[Example~6.4 on p.~318]{JoSa1998}).
Let $\gamma_{r}$ be the curve (clockwise orientation) on the $r$-$p_{r}$ plane (see Fig~\ref{fig:PeriodicOrbits}) defined by $p_{r}^{2} = \parentheses{ \tod{W_{r}}{r} }^{2}$, i.e.,
\begin{equation*}
  \frac{2r}{m}\, p_{r}^{2}
  + \frac{m\,\omega^{2}}{2}r
  + \frac{\hbar^{2}}{8 m r} = E_{r},
\end{equation*}
or
\begin{equation*}
  p_{r} = \pm \frac{m\,\omega}{2} \sqrt{
    -1 + \frac{\alpha}{r} - \frac{L^{2}}{2r^{2}}
  }.
\end{equation*}
Setting $p_{r} = 0$ yields $r = r_{\pm} \defeq (\alpha \pm \sqrt{ \alpha^{2} - 2L^{2} } )/2$.
Then, the action variable $I_{r}$ is calculated as follows:
\begin{align*}
  I_{r} &= \frac{1}{2\pi} \oint_{\gamma_{r}} p_{r}{\bf d}r
  = \frac{m\,\omega}{2\pi} \int_{r_{-}}^{r_{+}} \sqrt{
    -1 + \frac{\alpha}{r} - \frac{L^{2}}{2r^{2}}
  }\, dr
  \\
  &= \frac{E_{r}}{2\omega} - \frac{\hbar}{4}
  = \frac{r}{m\,\omega}\, p_{r}^{2}
    + \frac{m\,\omega}{4}r
    + \frac{\hbar^{2}}{16 m\,\omega\, r}
    - \frac{\hbar}{4}.
\end{align*}
The angle variable $\theta_{r}$ is then given by
\begin{equation*}
  \theta_{r} = \pd{}{I_{r}} \int \od{W_{r}}{r}\,dr
  = \tan^{-1}\brackets{
    \frac{ 4r^{2}(m^{2}\omega^{2} - 4p_{r}^{2}) - \hbar^{2} }{ 16 m\,\omega\, r^{2} p_{r} }
  }.
\end{equation*}

The (reduced) Hamiltonian $\overline{H}$ is then written in terms of the action variables as follows:
\begin{equation*}
  \overline{H} = \parentheses{ I_{1} + 2I_{r} + \frac{\hbar}{2} } \omega.
\end{equation*}
Then, the reduced dynamics on $\overline{\mathcal{M}}_{\hbar}$ is written as
\begin{equation*}
  \dot{\theta}_{1} = \pd{\overline{H}}{I_{1}} = \omega,
  \qquad
  \dot{\theta}_{r} = \pd{\overline{H}}{I_{r}} = 2\omega,
\end{equation*}
and $I_{1}$ and $I_{r}$ are constant.
Therefore, the reduced dynamics is now transformed to a periodic flow on the torus $\mathbb{T}^{2} = \mathbb{S}^{1} \times \mathbb{S}^{1} = \{ (\theta_{1}, \theta_{r}) \}$.

\subsection{Calculation of Geometric Phase}
\label{ssec:CalcOfGeomPhase}
Recall, from Section~\ref{sec:ReconstructionAndGeometricPhase}, that we may calculate the geometric phase achieved by a periodic motion of the reduced dynamics on $\overline{\mathcal{M}}_{\hbar}$.
The previous section revealed that the reduced dynamics is in fact periodic with period $T = 2\pi/\omega$; we have also obtained the curves traced by the periodic solution on the $q$-$p$ and $r$-$p_{r}$ planes.
These results enable us to calculate the geometric phase explicitly.
First recall from \eqref{eq:GeometricPhase} with \eqref{eq:omega-canonical} that we have
\begin{equation*}
  \Delta\phi_{\text{geom}} = -\iint_{D} ({\bf d}q \wedge {\bf d}p + {\bf d}r \wedge {\bf d}p_{r}),
\end{equation*}
where $D$ is any two-dimensional submanifold in $\overline{\mathcal{M}}_{\hbar}$ whose boundary is the periodic orbit $\Gamma \subset \overline{\mathcal{M}}_{\hbar}$, i.e., the curve $c: [0, T] \to \overline{\mathcal{M}}_{\hbar}$ defined by the reduced dynamics.
Then, the projections of the curve $\Gamma$ to the $q$-$p$ and $r$-$p_{r}$ planes are the curves $\gamma_{1}$ and $\gamma_{r}$ defined above, including the orientations (note that the clockwise orientations for $\gamma_{1}$ and $\gamma_{r}$ coincide with the direction of the dynamics on $\Gamma$).
Therefore, we have
\begin{equation*}
  \Delta\phi_{\text{geom}} = \oint_{\gamma_{1}} p\,{\bf d}q + 2 \oint_{\gamma_{r}} p_{r}{\bf d}r,
\end{equation*}
since the projection of $\Gamma$ to the $r$-$p_{r}$ plane gives two cycles of $\gamma_{r}$ for a single period $T = 2\pi/\omega$.
Using the expressions for $p$ and $p_{r}$ from the above subsections, we obtain
\begin{equation*}
 \Delta\phi_{\text{geom}} = \frac{2\pi E}{\omega} - \pi\hbar = E\,T - \pi\hbar,
\end{equation*}
which gives the following Aharonov--Anandan phase (note that the phase factor in \eqref{eq:chi} is $e^{i\phi/\hbar}$):
\begin{equation*}
  \Delta\parentheses{ \frac{\phi_{\text{geom}}}{\hbar} } = \frac{E\,T}{\hbar} - \pi,
\end{equation*}
and hence the total phase change is given by, using Eq.~\eqref{eq:TotalPhase}, 
\begin{equation*}
  \Delta\parentheses{ \frac{\phi_{\text{total}}}{\hbar} } = \Delta\parentheses{ \frac{\phi_{\text{dyn}} + \phi_{\text{geom}}}{\hbar} } = -\pi.
\end{equation*}
This implies that the corresponding wave function (see Eq.~\eqref{eq:chi}) flips ``upside down'' (just like a falling cat!) after one period, i.e.,
\begin{equation*}
  \iota \circ y(T) = -\iota \circ y(0)
  \quad\text{or}\quad
  \chi(y(T);x) = -\chi(y(0);x).
\end{equation*}

\section{Example~2: Semiclassical Tunneling Escape}
\label{sec:SemiclassicalTunneling}
For the above harmonic oscillator example, one does not observe quantum effects in the trajectory $q(t)$ of the particle, as the equations for the position $q$ and momentum $p$ coincide with the classical Hamiltonian system for the harmonic oscillator as the potential is quadratic (see Remark~\ref{rem:quadraticV}).

In order to observe quantum effects taken into account by the correction term in the potential \eqref{eq:avgV-asymptotic}, let us consider the following one-dimensional example with an anharmonic potential term from \citet{PrPe2000} (see also \citet{Pr2006}):
\begin{equation}
  \label{eq:anharmonicV}
  V(x) = \frac{1}{2}m\,\omega^{2}x^{2} + c\,x^{3}.
\end{equation}
Note that the asymptotic expansion \eqref{eq:avgV-asymptotic-spherical} of the potential term $\overline{\exval{V}}(q, b)$ terminates at the first-order of $\hbar$ and gives the exact value of $\overline{\exval{V}}(q, b)$:
\begin{equation*}
  \overline{\exval{V}}(q, b) = V(q) + \hbar\,\frac{1 + 6c\,q}{4b}.
\end{equation*}

Following \citet{PrPe2000}, we choose the parameters as follows: $m = 1$, $\omega = 1$, $c = 1/10$, and $\hbar = 1$ (to make quantum effects more prominent, although this is not quite in the semiclassical regime).

The initial condition is
\begin{equation}
  \label{eq:IC-tunneling}
  (q(0),p(0),a(0),b(0),\phi(0),\delta(0)) = (1,-1,0,1,-1,\ln(\pi)/4),
\end{equation}
where the value of $\delta(0)$ is chosen so that the Gaussian wave packet is normalized.
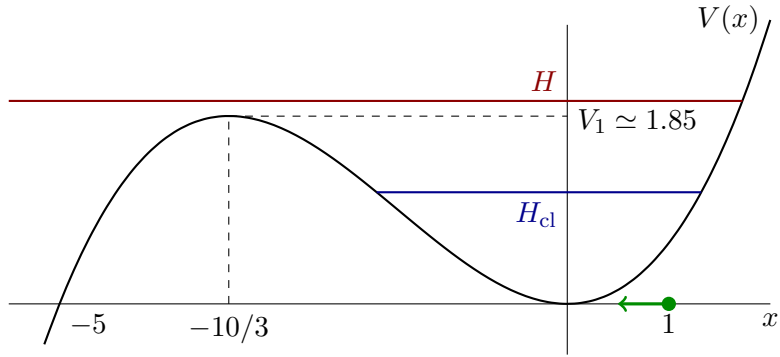
\begin{figure}[htbp]
  \begin{center}
    \begin{tikzpicture}[smooth, domain=-5.15:2, scale=1.35, samples=50]
      \draw[-] (-5.5,0) -- (2,0) node[below] {$x$};
      \draw[-] (0,-0.5) -- (0,2.75);
      \draw[thick] plot (\x,{0.5*\x*\x + 0.1*\x^3}) node[left] {$\DS V(x)$};
      \node[below right] at (-5,0) {$-5$};
      \node[below] at (1,0) {$1$};
      \fill [green!55!black] (1,0) circle (2pt);
      \draw[->, very thick, green!55!black] (1,0) -- (0.5,0);
      \node[below] at (-10/3,0) {$-10/3$};
      \draw[dashed] (-10/3,0) -- (-10/3,1.85);
      \draw[-, thick, blue!55!black] (-1.87667,1.1) -- (1.31935,1.1);
      \node[below left, blue!55!black] at (0,1.1) {$H_{\text{cl}}$};
      \draw[-, thick, red!55!black] (1.72458,2) -- (-5.5,2);
      \node[above left, red!55!black] at (0,2) {$H$};
      \draw[dashed] (-3.33,1.85) -- (0,1.85);
      \node[right] at (0,1.8) {$V_{1} \simeq 1.85$};
    \end{tikzpicture}
  \end{center}
  \caption{
    Potential \eqref{eq:anharmonicV} with $m = 1$, $\omega = 1$, $c = 1/10$; $H_{\text{cl}}$ is the classical Hamiltonian and $H$ is the semiclassical Hamiltonian~\eqref{eq:H-spherical} with initial condition~\eqref{eq:IC-tunneling}.
    $H_{\text{cl}} < V_{1} < H$ implies classical trajectory is trapped but semiclassical trajectory may escape through the potential barrier.
    The green dot and arrow on the $x$-axis indicate the initial position and velocity of the particle.
  }
  \label{fig:anharmonicV}
\end{figure}

\begin{figure}[htbp]
  \centering
  \subfigure[Phase portrait]{
    \includegraphics[width=.45\linewidth]{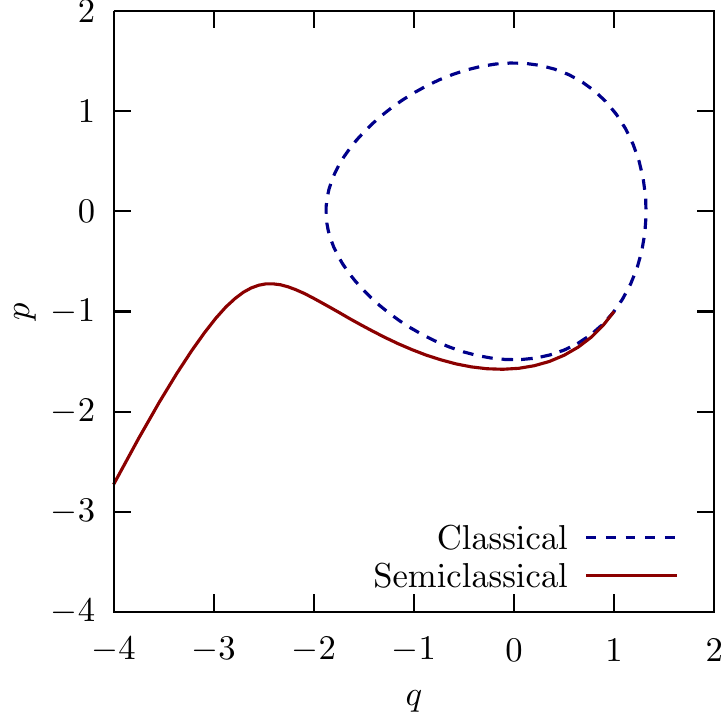}
  }
  \quad
  \subfigure[Time evolution of the position of the particle]{
    \includegraphics[width=.45\linewidth]{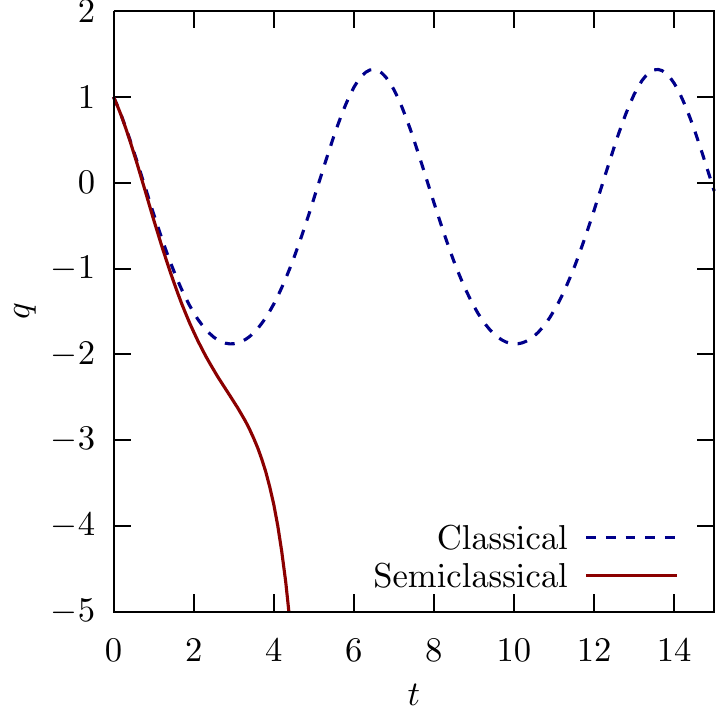}
  }
  \caption{Semiclassical Tunneling: The semiclassical solution escapes the potential well inside which the classical solution is trapped.}
  \label{fig:SemiclassicalTunneling}
\end{figure}

Figure~\ref{fig:anharmonicV} shows the shape of the potential as well as the values of the classical and semiclassical Hamiltonians $H_{\text{cl}}$ and $H$ with the above initial condition~\eqref{eq:IC-tunneling}.
The potential has a local maximum at $x = -10/3$ with $V_{1} \defeq V(-10/3) \simeq 1.85$, whereas $H_{\text{cl}} = p^{2}/2m + V(q) \simeq 1.1$ and $H = 2$ and so $H_{\text{cl}} < V_{1} < H$; this implies that the classical trajectory is trapped inside the potential well and undergoes a periodic motion, whereas the semiclassical trajectory may tunnel through the wall at $x = -10/3$.
Intuitively speaking, the variables $(q,p)$ in the semiclassical equations may ``borrow'' some energy from the variables $(a,b)$ (see the expression for the semiclassical Hamiltonian~\eqref{eq:H-spherical}) and therefore may have extra energy to climb over the wall.

Figure~\ref{fig:SemiclassicalTunneling} shows the phase portrait and the time evolution of the position of the particle for both the semiclassical and classical solutions with the same initial condition for $(q,p)$.
We used the variational splitting integrator of \citet{FaLu2006} (see also \citet[Section~IV.4]{Lu2008}) for the semiclassical solution and the St\"ormer--Verlet method~\cite{Ve1967} for the classical solution; the time step is $0.1$ in both cases.
(It is perhaps worth mentioning that the variational splitting integrator is a natural extension of the St\"ormer--Verlet method in the sense that it recovers the St\"ormer--Verlet method as $\hbar \to 0$~\cite{FaLu2006}.)
We observe that the semiclassical ``particle'' indeed escapes from the potential well, whereas the classical solution is trapped inside the potential well.
Compare them with Figures~2 and 3 of \citet{PrPe2000}: Our semiclassical solution seems to be almost identical to the solution of their second-order Quantized Hamiltonian Dynamics (QHD), which is shown to approximate the solution of the Schr\"odinger equation much better than the classical solution does~\cite{PrPe2000}.
In fact, as mentioned in \cite{PrPe2000} and \cite{Pr2006}, the second-order QHD yields equations similar to those of \citet{He1975a} with a correction term to the classical potential.
It is likely that the second-order QHD is identical to \eqref{eq:Heller-asymptotic}, although the relationship between them is not so clear to us as the two approaches are quite different in spirit: The Gaussian wave packet approach uses the Schr\"odinger picture, whereas the QHD employs the Heisenberg picture.
It is an interesting future work to bridge the gap between the two.

\section{Conclusion}
We gave a symplectic-geometric account of Heller's semiclassical Gaussian wave packet dynamics that builds upon on a series of works by Lubich and his collaborators.
Our point of view is helpful in understanding how semiclassical dynamics inherits the geometric structures of quantum dynamics.
Particularly, the geometry behind the symplectic reduction and reconstruction of semiclassical dynamics is inherited from quantum dynamics in a natural way.
We also derived an asymptotic formula for the expected value of the potential to approximate the potential terms appearing in the system of equations for semiclassical dynamics.
The asymptotic formula not only naturally generalizes Heller's approximation but also indicates that it is crucial to couple the equations for the classical position and momentum variables $q$ and $p$ with those of the other quantum variables, thereby justifying our point of view of regarding the whole system as a Hamiltonian system.

\section*{Acknowledgments}
We would like to thank Christian Lubich, Peter Miller, Peter Smereka, and the referees for their helpful comments and discussions.
A connection with coherent states along the lines of Section~\ref{ssec:Littlejohn} was brought to our attention by David Meier, Richard Montgomery, and one of the referees.
This material is based upon work supported by the National Science Foundation under the Faculty Early Career Development (CAREER) award DMS-1010687, the FRG grant DMS-1065972, and the AMS--Simons Travel Grant.

\appendix

\section{Proof of Proposition~\ref{prop:asymptoticV}}
\label{sec:asymptoticV-proof}
First take the ball $B_{\eps}(q)$ and then split the integral $F_{\hbar}$ in Eq.~\eqref{eq:F_hbar} as follows:
\begin{equation*}
  F_{\hbar}(q, \mathcal{B}) = \int_{B_{\eps}(q)} e^{R(x)/\hbar}\, V(x)\,dx +  \int_{\R^{d}\backslash B_{\eps}(q)} e^{R(x)/\hbar}\, V(x)\,dx.
\end{equation*}
Then, since $V$ is assumed to be smooth, we obtain the asymptotic expansion \eqref{eq:asymptoticV} by applying the standard result of Laplace's method (see, e.g., \citet[Section~3.7]{Mi2006}) to the first term.
We need an estimate of the second term to justify the expansion.
Introducing the variable $\xi \defeq x - q$ and defining
\begin{equation*}
  \tilde{R}(\xi) \defeq R(q + \xi) = -\xi^{T}\mathcal{B}\xi
  \quad \text{and} \quad
  \tilde{V}(\xi) \defeq V(q + \xi), 
\end{equation*}
we have
\begin{align*}
  \int_{\R^{d}\backslash B_{\eps}(q)} e^{R(x)/\hbar}\, V(x)\,dx
  &= \int_{\R^{d}\backslash B_{\eps}(0)} e^{\tilde{R}(\xi)/\hbar}\, \tilde{V}(\xi)\,d\xi
  \\
  &\le \parentheses{ \int_{\R^{d}\backslash B_{\eps}(0)} e^{2(1-\sigma)\tilde{R}(\xi)/\hbar}\,d\xi }^{1/2}
  \parentheses{ \int_{\R^{d}\backslash B_{\eps}(0)} \brackets{ e^{\sigma\tilde{R}(\xi)/\hbar}\, \tilde{V}(\xi) }^{2}\,d\xi }^{1/2},
\end{align*}
where we used the Cauchy--Schwarz inequality.
The second term is bounded by assumption.
To evaluate the first term, we introduce the new variable $\eta = \mathcal{Q}^{T} \xi$; then the exponent simplifies to
\begin{equation*}
  2(1-\sigma)\tilde{R}(\mathcal{Q}\eta) = -2(1-\sigma)\,\eta^{T}\mathcal{Q}^{T}\mathcal{B}\mathcal{Q} \eta = -\sum_{k=1}^{d} \beta_{k} \eta_{k}^{2},
\end{equation*}
where we set $\beta_{k} \defeq 2(1-\sigma)b_{k}$, which is positive.
Thus, we have
\begin{align*}
  \int_{\R^{d}\backslash B_{\eps}(0)} e^{2(1-\sigma)\tilde{R}(\xi)/\hbar}\,d\xi
  &= \int_{\R^{d}\backslash B_{\eps}(0)} e^{-\sum_{k=1}^{d} \beta_{k} \eta_{k}^{2}/\hbar}\,d\eta
  \\
  &\le \int_{\R^{d}\backslash C_{\eps/\sqrt{d}}} e^{-\sum_{k=1}^{d} \beta_{k} \eta_{k}^{2}/\hbar}\,d\eta,
  \\
  &= \prod_{k=1}^{d} \int_{|\eta_{k}| \ge \eps/\sqrt{d}} e^{-\beta_{k} \eta_{k}^{2}/\hbar}\,d\eta_{k},
\end{align*}
where $C_{\eps/\sqrt{d}}$ is the hypercube defined by 
\begin{equation*}
  C_{\eps/\sqrt{d}} \defeq \setdef{ \eta \in \R^{d} }{ |\eta_{k}| < \frac{\eps}{\sqrt{d}} \text{ for } k = 1, \dots, d },
\end{equation*}
which is clearly contained in $B_{\eps}(0)$.
Writing $\eps_{d} \defeq \eps/\sqrt{d}$ for shorthand, the Cauchy--Schwarz inequality gives
\begin{align*}
  \int_{|\eta_{k}| \ge \eps_{d}} e^{-\beta_{k} \eta_{k}^{2}/\hbar}\,d\eta_{k}
  &= 2 \int_{\eps_{d}}^{\infty} e^{-\beta_{k} \eta_{k}^{2}/\hbar}\,d\eta_{k}
  \\
  &= 2e^{\beta_{k}\eps_{d}^{2}/\hbar} \int_{\eps_{d}}^{\infty} e^{-\beta_{k} (\eta_{k} - \eps_{d})^{2}/\hbar}\, e^{-2\beta_{k} \eta_{k} \eps_{d}/\hbar}\,d\eta_{k}
  \\
  &\le 2e^{\beta_{k}\eps_{d}^{2}/\hbar} \parentheses{ \int_{\eps_{d}}^{\infty} e^{-2\beta_{k} (\eta_{k} - \eps_{d})^{2}/\hbar}\,d\eta_{k} }^{1/2}
  \parentheses{ \int_{\eps_{d}}^{\infty} e^{-4\beta_{k} \eta_{k} \eps_{d}/\hbar}\,d\eta_{k} }^{1/2}
  \\
  &= \parentheses{ \frac{d \pi}{8 \eps^{2}} }^{1/4} \parentheses{ \frac{\hbar}{ \beta_{k} } }^{3/4}\, e^{-\beta_{k} \eps^{2}/(d\hbar)}.
  \\
  &= \parentheses{ \frac{d \pi}{8 \eps^{2}} }^{1/4} \parentheses{ \frac{\hbar}{ 2(1-\sigma)b_{k} } }^{3/4}\, e^{-\beta_{k} \eps^{2}/(d\hbar)}.
\end{align*}
Therefore,
\begin{equation*}
  \int_{\R^{d}\backslash B_{\eps}(0)} e^{\tilde{R}(\xi)/\hbar}\,d\xi
  \le \parentheses{ \frac{d \pi}{8 \eps^{2}} }^{d/4} \parentheses{ \frac{\hbar^{d}}{ [2(1-\sigma)]^{d}\det\mathcal{B} } }^{3/4}\,
  \exp\parentheses{ -\frac{2\eps^{2}(1-\sigma)\tr\mathcal{B}}{d\hbar} }
  = o(\hbar^{p})
\end{equation*}
as $\hbar \to 0$ for any real $p$, since the above exponential term is dominated by any real power of $\hbar$.
Therefore,
\begin{equation*}
  \int_{\R^{d}\backslash B_{\eps}(q)} e^{R(x)/\hbar}\, V(x)\,dx = o(\hbar^{p})
  \quad \text{as} \quad \hbar \to 0
\end{equation*}
for any real $p$ as well, and so the above integral has no contribution to the asymptotic expansion.

\bibliography{SympSemiCl}
\bibliographystyle{plainnat}

\end{document}